\documentclass[reqno, 10pt]{amsart}

\usepackage[top=4cm, bottom=3cm, left=2.8cm, right=2.8cm]{geometry}

\usepackage{enumerate}

\newtheorem{thm}{Theorem}

\newtheorem{lem}[thm]{Lemma}
\newtheorem{prop}[thm]{Proposition}
\theoremstyle{definition}
\newtheorem{defn}[thm]{Definition}
\theoremstyle{remark}

\numberwithin{equation}{section}
\numberwithin{thm}{section}

\newcommand{\infspec}{{\rm inf\ spec\ }}
\newcommand{\x}{{\textbf{x}}}
\newcommand{\N}{{\mathcal{N}}}
\newcommand{\M}{{\mathcal{M}}}

\begin{document}

\title{A Lower Bound on the Ground State Energy of Dilute Bose Gas}

\author[Ji Oon Lee]{Ji Oon Lee}

\address[Ji Oon Lee]{Department of Mathematics, Harvard University, Cambridge, MA 02138, USA}

\email{jioon@math.harvard.edu}

\author[Jun Yin]{Jun Yin}

\address[Jun Yin]{Department of Mathematics, Harvard University, Cambridge, MA 02138, USA}

\email{jyin@math.harvard.edu}

\begin{abstract}
Consider an N-Boson system interacting via a two-body repulsive short-range potential $V$ in a three dimensional box $\Lambda$ of side length $L$. We take the limit $N, L \to \infty$ while keeping the density $\rho = N / L^3$ fixed and small. We prove a new lower bound for its ground state energy per particle $$\frac{E(N, \Lambda)}{N} \geq 4 \pi a \rho [ 1 - O( \rho^{1/3} |\log \rho|^3 ) ],$$ as $\rho \to 0$, where $a$ is the scattering length of $V$.
\end{abstract}

\maketitle \thispagestyle{empty}
\thispagestyle{headings}

\section{Introduction}

The properties of the Bose gas have been studied by many authors \cite{LHY, D, LY2, L1}, and since the first experimental observation of Bose-Einstein Condensation in 1995 \cite{K}, interest in low temperature Bose gases are renewed \cite{LY1, LY3, LS, LS1, LSY, S, SY}. One of the most well-known properties of Bose gas is its ground state energy in the dilute limit. In this low density limit, the leading term of the ground state energy per particle is $4 \pi a \rho$, where $a$ is the scattering length of two-body interaction potential and $\rho$ is the density. The upper bound for it was first rigorously proved by Dyson \cite{D}, and the lower bound was obtained by Lieb and Yngvason \cite{LY1}. It is also proved that the leading term is the same in some cases where the interaction potential is not purely non-negative \cite{L, Y}.

Lee and Yang \cite{LY, LY2} first predicted the second order correction to this leading term, which is given by
\begin{equation}
4 \pi a \rho \left[ 1 + \frac{128}{15 \sqrt{\pi}} (\rho a^3)^{\frac{1}{2}} + \cdots \right].
\end{equation}
This calculation used pseudopotential method and binary collision expansion method \cite{LHY}. Another derivation was later given by Lieb \cite{L1} using a self-consistent closure assumption for the hierarchy of correlation functions.

The upper bound of this second order correction is recently obtained in Yau and Yin's work \cite{YY}. The matching lower bound, however, has not been proved yet, and in fact, the best result so far is $4 \pi a \rho ( 1 - C \rho^{1/17})$, which was first derived by Lieb and Yngvason \cite{LY1}. (See also \cite{LSSY} for details.) It also should be mentioned that this second order correction is proved for the high density and weak coupling regime recently by Giuliani and Seiringer \cite{GS}.

Lieb and Yngvason \cite{LY1} used the ``cell method'' to find the lower bound for the ground state energy. In this method, they first converted the interaction potential into a soft potential with the expense of kinetic energy. Unlike Dyson \cite{D}, however, they did not sacrifice all the kinetic energy but kept a small portion of kinetic energy to apply Temple's inequality \cite{T}. Then, after dividing a large box into smaller cubic cells, perturbation theory was applied in each small cubic cell. In this way, the correct leading term could be obtained.

We improve this approach with several new ideas including a new method called `box doubling method.' Some notable advantages of our methods are as follows:
\begin{itemize}
\item We use a combination of perturbation theory and Temple inequality. Also, we are able to preserve essentially the full kinetic energy after the replacement of the the singular potential by a smooth one.

\item Each small cell can be enlarged so that it contains enough number of particles in average.
\end{itemize}

In this paper we show that the error term becomes $- C \rho^{1/3} |\log \rho|^3$ with the new method. Our main result is stated in Theorem \ref{main theorem}.

This paper is organized as follows: In Section 2, we state our main theorem, introduce key lemmas, and prove the theorem. In section 3, we show new type of `cell method', which is equivalent to proving some of the key lemmas. In section 4, we introduce `box doubling method', which enlarges each cell, hence makes it contain more particles in average. Some technical estimates can be found in section 5.
\section{Main result and Key lemmas}

\subsection{Definition of the system and main result}
Let $V$ be a non-negative, smooth, spherically symmetric, and compactly supported potential, satisfying $V(x) = 0$ if $|x|>R_0$ for $x \in \mathbb{R}^3$. The zero energy scattering solution $\varphi$ satisfies
\begin{equation}
\left( -\Delta + \frac{1}{2} V(x) \right) \varphi(x) = 0.
\end{equation}
We normalize $\varphi$ so that $\lim_{|x| \to \infty} \varphi(x) = 1$. The scattering length $a$ of $V$ is defined by
\begin{equation}
a := \lim_{|x| \to \infty} |x| \left( 1  - \varphi(x) \right)
\end{equation}
as in \cite{LY1}.

Hamiltonian of the Bose gas system of $N$ particles in a torus $\Lambda$ of side length $L$ is given by
\begin{equation} \label{definition of H}
H_N := -\sum_{j=1}^N \Delta_j + \sum_{i<j}^N V(x_i -x_j).
\end{equation}
We will consider the thermodynamic limit when $N$ and $L$ approach infinity with the density $\rho = N/L^3$ fixed. Later, we will also let $\rho \to 0$. We will use the notation $\alpha \ll \beta$ if $\alpha / \beta = O(\rho^{c})$ or $\alpha / \beta = O(|\log \rho|^{-c})$ for some $c>0$, and $\alpha \sim \beta$ if $\alpha / \beta = O(1)$. For simplicity, we will also use the notation
\begin{equation}
\x_n = (x_1, x_2, \cdots, x_n), \;\;\; \x = \x_N.
\end{equation}
Throughout the paper, $C$ denotes a constant independent of $\rho$.

\begin{defn}[The ground state energy]
For given Hamiltonian
\begin{equation}
H_n := -\sum_{j=1}^n \Delta_j + \sum_{i<j}^n V(x_i -x_j)
\end{equation}
in the three dimensional torus $\Lambda$, its ground state energy, $E(n, \Lambda)$ is defined to be
\begin{equation}
E(n, \Lambda) := \inf_{\| \psi \|_2 =1} \Big[ \sum_{j=1}^n \int_{\Lambda^N} |\nabla_j \psi(\x)|^2 d\x + \sum_{i<j}^n \int_{\Lambda^N} V(x_i - x_j) |\psi(\x)|^2 d\x \Big].
\end{equation}
\end{defn}

The main result of this paper is the following theorem.
\begin{thm} \label{main theorem}
Let $V$ be non-negative, smooth, spherically symmetric, and compactly supported potential, satisfying $V(x) = 0$ if $|x|>R_0$, whose scattering length is $a$. Then, there exists a constant $C_0 > 0$ such that
\begin{equation}
\mathop{\lim_{N, L \to \infty}}_{N/L^3 = \rho} \frac{E(N, \Lambda)}{N} \geq 4 \pi a \rho ( 1 - C_0 \rho^{\frac{1}{3}} |\log \rho|^3 )
\end{equation}
as $\rho \to 0$.
\end{thm}

\subsection{Notations}

Let $e_0 (\kappa)$ and $(1-\tau(\kappa, x))$ be the lowest Neumann eigenvalue and eigenfunction of $(-\Delta + \frac{1}{2} V)$ on the ball of radius $\kappa \gg R_0$, i.e.,
\begin{equation}
(-\Delta + \frac{1}{2} V(x)) (1-\tau(\kappa, x)) = e_0 (\kappa) (1-\tau(\kappa, x))
\end{equation}
with the boundary conditions
\begin{equation}
\tau(\kappa, x)=0, \;\;\; \partial \tau(\kappa, x) = 0, \;\;\; \text{ if } |x| = \kappa.
\end{equation}
Note that $\tau$ is spherically symmetric, since $V$ is spherically symmetric. Some properties of $\tau$ and $e_0$ are collected in Section \ref{properties}, most notably,
\begin{equation}
\frac{3a}{\kappa^3} \leq e_0 \leq \frac{3a}{\kappa^3} ( 1 + \frac{C}{\kappa} ).
\end{equation}

Using $\tau$, we define a function $W_j$ that will be used for our proof, which will approximate the ground state of $-\Delta_j + \sum_{i: i \neq j} V(x_i -x_j)$ when $\x$ is fixed except $x_j$. Let $t_{ij}$ be the half of the distance of $x_i$ to its nearest particle other than $x_j$, i.e.,
\begin{equation}
t_{ij} = \frac{1}{2} \min_{k: k \neq i, j} |x_i - x_k|.
\end{equation}
We introduce a particle triple cutoff function $F_{ij}(\x)$ at a length scale $\ell_0$ defined by
\begin{equation}
F_{ij}(\x) = 
	\begin{cases}
	1 & \text{ if } t_{ij} > \ell_0 \\
	0 & \text{ otherwise }
	\end{cases}.
\end{equation}
We also introduce a particle triple cutoff function at a length scale $\ell_{-1} \ll \ell_0$, $G_{ij}$, which is defined to be
\begin{equation}
G_{ij}(\x) = 
	\begin{cases}
	1 & \text{ if } t_{ij} > \ell_{-1} \\
	0 & \text{ otherwise }
	\end{cases}.
\end{equation}
Here, $\ell_0$ and $\ell_{-1}$ are fixed and satisfy $\ell_0 \gg \ell_{-1} \gg R_0$. Note that $t_{ij} \neq t_{ji}$, and in fact, $t_{ij}$ does not depend on $x_j$. The same holds for $F_{ij}$ and $G_{ij}$.

Let us extend the definition of $\tau$ so that $\tau(\kappa, x) = 0$ if $|x| > \kappa$. We define $W_j$ by
\begin{equation} \label{definition of W}
W_j (\x) = 1 - \sum_{i: i \neq j} [F_{ij}(\x) \tau(\ell_0, x_i - x_j) + (1-F_{ij}(\x)) G_{ij}(\x) \tau(t_{ij}, x_i - x_j)].
\end{equation}
Our interpretation of this wave function $W_j$ is as follows: When $x_i$ is the nearest particle of $x_j$, $W_j$ becomes $1 - \tau(\ell_0, x_i - x_j)$ or $1 - \tau(t_{ij}, x_i - x_j)$, depending on $t_{ij}$. Physically, it corresponds to a situation in which the particle $x_j$ is moving and all the other particles are fixed. Any particle $x_i$ other than $x_j$ has its own `effect' on $x_j$, and the radius of this effect is given by $\ell_0$ or $t_{ij}$, provided that $t_{ij} \geq \ell_{-1}$. $x_j$ is affected only by its nearest particle. 

The definition of $W_j$ shows how we can handle the case when three or more particles are close to each other; when a particle $x_k$ is close to $x_i$ and $x_j$, we do not ignore interaction between $x_i$ and $x_j$ but assume that the range of interaction shrinks accordingly. Here, the range of interaction is $t_{ij}$, the size of support of $\tau$.

We will use length scales from now on such that
\begin{equation} \label{optimized exponents}
\ell_{-1} \sim \rho^{-\frac{2}{9}} |\log \rho|^{\frac{2}{3}}, \;\;\; \ell_0 \sim \rho^{-\frac{1}{3}} |\log \rho|^{-\frac{1}{3}}, \;\;\; \ell_1 \sim \rho^{-\frac{1}{3}} |\log \rho|^{\frac{1}{3} + \eta}, \;\;\; \ell_2 \sim \rho^{-\frac{4}{9}} |\log \rho|^{-\frac{2}{3}}, \;\;\; \epsilon = \rho^{\frac{1}{3}} |\log \rho|^3.
\end{equation}
Here, $\eta$ is a small positive number such that $0 < \eta < 1/15$. We choose the parameters properly to satisfy
\begin{equation}
L = h' \ell_2, \;\;\; \ell_2 = 2^h \ell_1,
\end{equation}
where $h, h'$ are integers.

We will prove in Section \ref{properties} that there exists a constant $c_0 < 1$ such that
\begin{equation}
W_j (\x) \geq 1 - c_0 >0.
\end{equation}
Other properties of $W_j$ will be proved in Section \ref{properties} as well.

\subsection{Basic strategy of the proof}

To prove the main theorem, we use a strategy that consists of four steps:

\begin{enumerate}
\item For a given Bosonic wave function $\Psi \in L^2 (\Lambda^N)$, we let $\Phi_j := W_j^{-1} \Psi$ for $j=1, 2, \cdots, N$. Here, $\Phi_j$ is well defined, since $W_j > 0$. Using $W_j$, we convert the hard potential $V$ into a soft potential $q$ with the expense of a portion of local kinetic energy:

\begin{lem} \label{energy estimate}
Let
\begin{equation}
q(\kappa, x) := e_0 (\kappa) 1(|x| \leq \kappa)
\end{equation}
and
\begin{equation}
q_{ij} := F_{ij} q(\ell_0, x_i - x_j) + (1- F_{ij}) G_{ij} q(t_{ij}, x_i - x_j).
\end{equation}
Let $\x = (x_1, x_2, \cdots, x_N)$. Then, for any $\Psi(\x) \in L^2 (\Lambda^N)$,
\begin{eqnarray}
&& \langle \Psi, H_N \Psi \rangle \geq \sum_{j=1}^N \int_{\Lambda^N} |W_j|^2 |\nabla_j \Phi_j(\x)|^2 d\x + \sum_{i \neq j}^N \int_{\Lambda^N} q_{ij}(\x) |\Psi(\x)|^2 d\x,
\end{eqnarray}
where $\Phi_j(\x) := W_j^{-1} \Psi(\x)$ for $j=1, 2, \cdots, N$.

\end{lem}

Note that, though $q_{ij}$ is the soft potential we want to use throughout the paper, it is not symmetric, i.e., $q_{ij} \neq q_{ji}$. The kinetic energy term here contains $|W_j|^2 |\nabla_j \Phi_j(\x)|^2$, which is associated with an operator $T_j = - W_j^{-1} \nabla_j W_j^2 \nabla_j W_j^{-1}$. To use this operator $T_j$, we need to impose boundary condition. Noting that Neumann boundary condition gives the lowest energy for a Laplacian operator $-\Delta_j$, we introduce the following condition on $T_j$ that corresponds to the Neumann boundary conditions on $-\Delta_j$:

\begin{defn} \label{condition N}
Let $U \subset \Lambda$ be a box. 
\begin{enumerate}
\item Assume that $x_1, \cdots, \widehat{x_j}, \cdots, x_N$ are fixed. We say that $\psi(x_j)$ satisfies `$W_j$-Neumann boundary condition' in $U$ when $\psi(x_j) \in L^2 (U)$ and $W_j^{-1} \psi(x_j)$ satisfies Neumann boundary condition in $U$.

\item Let $j = 1, 2, \cdots, n$. Assume that $x_{n+1}, \cdots, x_N$ are fixed. We say that $\psi(\x_n)$ satisfies `$W_j$-Neumann boundary condition' in $U$ when $\psi(\x_n) \in L^2 (U^n)$ and, for any fixed $x_1, \cdots, \widehat{x_j}, \cdots, x_n$, 
\begin{equation}
\psi_j(x_j) := \psi(\x_n)
\end{equation}
satisfies $W_j$-Neumann boundary condition in $U$ as defined above.
\end{enumerate}
\end{defn}

When we prove lemmas containing $|W_j|^2 |\nabla_j \Phi_j(\x)|^2$ such as Lemma \ref{small cell estimate} by using $T_j$, we will only consider cases where $\Phi_j(\x)$ satisfies $W_j$-Neumann boundary condition, which can be justified as follows:
For a given $\psi(\x_n) \in L^2 (U^n)$, we can find $\widetilde{\psi}(\x_n) \in L^2 (U^n)$ such that $\widetilde{\psi}$ satisfies $W_j$-Neumann boundary condition in $U$ for $j=1, 2, \cdots, n$, and the difference between the kinetic energies of $\psi$ and $\widetilde{\psi}$ is negligible. To see this, we first change $\psi$ into $\widetilde{\psi}$ so that $\nabla_j (W_j^{-1} \psi(x_j))=0$ at the boundary and $\psi$ differs from $\widetilde{\psi}$ only near the boundary. If this change was small enough and made sufficiently near the boundary, then, even though $\nabla_j \psi$ might have been changed a lot, the difference between $\nabla_j \psi$ and $\nabla_j \widetilde{\psi}$ could be arbitrarily small after integration. This results in the negligible difference between the kinetic energies of $\psi$ and $\widetilde{\psi}$.

\item To use perturbation theory, we divide the torus $\Lambda$ into small cubic cells. We need to replace the soft potential $q_{ij}$ by another soft potential $\widetilde{q_{ij}}$ so that particles in different boxes do not interact via this soft potential. If we just ignore interactions beyond the boundaries of cubic cells, however, it will lower the energy significantly. To resolve this problem, we shift the origin of division, $u$, continuously and take an average of energy with respect to $u$. More precisely, we will consider the following:

Let $u \in [-\ell_1/2, \ell_1/2)^3$ and let $G := \ell_1 \mathbb{Z}^3 \cap \Lambda$. For $\lambda \in G$, we let $\Lambda_{u \lambda}$ be a cubic box of side length $\ell_1$ centered at $(u + \lambda)$. Here, $u$ corresponds to the origin of the grid that divides $\Lambda$ into small boxes $\Lambda_{u \lambda}$. This technique leads us to the following lemma:

\begin{lem} \label{decomposition}
Consider a three dimensional grid with the origin $u \in \Lambda$. Assume that this grid divides $\Lambda$ into small cubic cells of side length $\ell_1 \sim \rho^{-\frac{1}{3}} |\log \rho|^{\frac{1}{3} + \eta}$, $0 < \eta < 1/15$ such that
\begin{equation}
\Lambda = \bigcup_{\lambda} \Lambda_{u \lambda}, \;\;\; \Lambda_{u \lambda_i} \cap \Lambda_{u \lambda_j} = \emptyset \text{ if } \lambda_i \neq \lambda_j,
\end{equation}
where $\Lambda_{u \lambda}$ is a cubic cell of side length $\ell_1$ and $\lambda$ is an index for those cubic cells.

For given $u$, let $\partial \Lambda_{u \lambda}$ be the boundary of $\Lambda_{u \lambda}$ and $d(x_i, \partial \Lambda_{u \lambda})$ the distance between $x_i$ and $\partial \Lambda_{u \lambda}$ in $\Lambda$. Define
\begin{equation}
\widetilde{q_{ij}} :=
	\begin{cases}
		q_{ij} & \text{ if } d(x_i, \partial \Lambda_{u \lambda}) > 2 \ell_0 \text{ for any } \lambda \\
		0 & \text{otherwise}
	\end{cases}
\end{equation}
Then, for any $\Psi(\x) = \Psi(x_1, x_2, \cdots, x_N) \in L^2 (\Lambda^N)$,
\begin{eqnarray}
&& \sum_{i \neq j}^N \int_{\Lambda^N} q_{ij}(\x) |\Psi(\x)|^2 d\x \geq \inf_u (\frac{\ell_1}{\ell_1 - 4\ell_0})^3 \sum_{i \neq j}^N \int_{\Lambda^N} \widetilde{q_{ij}}(\x) |\Psi(\x)|^2 d\x.
\end{eqnarray}
\end{lem}

\item In each small cubic cell of side length $\ell_1$, we use perturbation theory with half of kinetic energy to estimate the energy in the small cell in terms of the number of particles in it. When the small cell contains too many particles, however, perturbation loses its validity. In this case, a very small portion of energy from the main Hamiltonian $H_N$ will contribute the energy we need and it is sufficient in order to estimate the lower bound for the ground state energy. More precisely, we can prove the following lemma:

\begin{lem} \label{small cell estimate}
Let $x_1, x_2, \cdots, x_N \in \Lambda$. Let $B \subset \Lambda$ be a box of side length $\ell_1$. Define $E(n, B)$ as the ground state energy of $n$ particle system in $B$ with interaction $V$, i.e.,
\begin{equation}
E(n, B) := \infspec \{ -\sum_{j=1}^{n} \Delta_j + \sum_{i < j}^{n} V(x_i - x_j) \}
\end{equation}
with the Neumann boundary conditions. Let
\begin{equation} \label{definition of f}
f(n) :=
	\begin{cases}
	n(n-1) & \text{ if } n \leq 2 \rho \ell_1^3 \\
	(4 \rho \ell_1^3 - 1) n - (2 \rho \ell_1^3)^2 & \text{ if } n > 2 \rho \ell_1^3
	\end{cases}.
\end{equation}
Let $\epsilon = \rho^{\frac{1}{3}} |\log \rho|^3$ as in \eqref{optimized exponents}. Fix $x_{n+1}, x_{n+2}, \cdots, x_N$ outside $B$ and let $\x_n = (x_1, x_2, \cdots, x_n)$. Then, for any $\psi(\x_n) \in L^2 (B^n)$ with $\phi_j := W_j^{-1} \psi$ for $j=1, 2, \cdots, n$, we have
\begin{eqnarray}
&& \epsilon E(n, B) \| \psi \|_2^2 + \frac{1}{2} \sum_{j=1}^n \int_{B^n} |W_j|^2 |\nabla_j \phi_j(\x_n)|^2 d\x_n + (\frac{\ell_1}{\ell_1 - 4\ell_0})^3 \sum_{i \neq j}^{n} \int_{B^n} \widetilde{q_{ij}} |\psi(\x_n)|^2 d\x_n \nonumber \\
&\geq& f(n) (1 - C \epsilon) \frac{4 \pi a}{\ell_1^3} \| \psi \|_2^2. \label{potential term in small cell estimate}
\end{eqnarray}
\end{lem}

\item In the cell method in step (3), the number of pairs in a small cubic cell is not $n^2/2$ but $n(n-1)/2$. This gives an additional factor $(1 - \rho^{-1} \ell_1^{-3})$, which is the ratio between $n(n-1)$ and $n^2$. This error factor becomes quite problematic because $\rho^{-1} \ell_1^{-3} \sim |\log \rho|^{-1 - 3 \eta} \gg \epsilon$. To handle this problem, we merge $2^{3h}$ adjacent small cubic cells of side length $\ell_1$ into a larger cubic box of side length $\ell_2$. (Recall that we chose $\ell_2 = 2^h \ell_1$.) For a technical reason, we only merge two cells at a time, and we apply perturbation theory to achieve a similar result to Lemma \ref{small cell estimate} for this `doubled box.' When we get to the cell of side length $\ell_2$, we can obtain the following lemma:

\begin{lem} \label{large cell estimate}
Let $x_1, x_2, \cdots, x_N \in \Lambda$. Let $\Lambda_{\ell_2} \subset \Lambda$ be a cubic box of side length $\ell_2 \sim \rho^{-\frac{4}{9}} |\log \rho|^{-\frac{2}{3}}$. Divide $\Lambda_{\ell_2}$ into $2^{3h}$ smaller cubic cells of side length $\ell_1$ such that $\ell_1 = 2^{-h} \ell_2$. Call those small cubic cells $B_1, B_2, \cdots, B_{2^{3h}}$.
Define $f$ as in \eqref{definition of f} and let
\begin{equation}
\widetilde{f}(t) :=
	\begin{cases}
	t(t-1) & \text{ if } t \leq \rho \ell_2^3 \\
	(2 \rho \ell_2^3 - 1) t - (\rho \ell_2^3)^2 & \text{ if } t > \rho \ell_2^3
	\end{cases}.
\end{equation}
Fix $x_{n+1}, x_{n+2}, \cdots, x_N$ outside $\Lambda_{\ell_2}$ and let $\x_n = (x_1, x_2, \cdots, x_n)$. Define $\N(B_k)$ as
\begin{equation}
\N(B_k) := \sum_{j=1}^{n} 1(x_j \in B_k).
\end{equation}
Then, for any $\psi(\x_n) \in L^2 (\Lambda_{\ell_2}^n)$ with $\phi_j = W_j^{-1} \psi$ for $j=1, 2, \cdots, n$,
\begin{eqnarray}
&& \big( \frac{1}{2} - \epsilon \big) \sum_{j=1}^n \int_{\Lambda_{\ell_2}^n} |W_j|^2 |\nabla_j \phi_j(\x_n)|^2 d\x_n + \frac{4 \pi a}{\ell_1^3} \sum_{k=1}^{2^{3h}} \int_{\Lambda_{\ell_2}^n} f(\N(B_k)) |\psi(\x_n)|^2 d\x_n \nonumber \\
&\geq& \widetilde{f}(n) (1 - C \rho^{\frac{1}{3}} |\log \rho|) \frac{4 \pi a}{\ell_2^3} \| \psi \|_2^2. \label{result of large cell estimate}
\end{eqnarray}

\end{lem}

\end{enumerate}

Proofs of the key lemmas above are given in Section \ref{key lemmas} and Section \ref{doubling a box}.

\subsection{Proof of main theorem}

\begin{proof}[Proof of Theorem \ref{main theorem}]

Assume that $\rho$ is sufficiently small. Lemma \ref{energy estimate} shows that, for any $\Psi \in L^2 (\Lambda^N)$,
\begin{eqnarray}
\langle \Psi, H_N \Psi \rangle \geq \epsilon \langle \Psi, H_N \Psi \rangle + (1 - \epsilon) \Big( \sum_{j=1}^N \int_{\Lambda^N} |W_j(\x)|^2 |\nabla_j \Phi_j(\x)|^2 d\x + \sum_{i \neq j}^N \int_{\Lambda^N} q_{ij} |\Psi(\x)|^2 d\x \Big), \label{result after energy estimate}
\end{eqnarray}
where $\Phi_j = W_j^{-1} \Psi$ and $\epsilon = \rho^{\frac{1}{3}} |\log \rho|^3$. 

Recall that in Lemma \ref{decomposition} we divided $\Lambda$ into small cubic cells $\Lambda_{u \lambda}$ of side length $\ell_1$, where $u$ is the origin of division and $\lambda \in G$ is an index for those cells. From Lemma \ref{decomposition}, we get
\begin{eqnarray}
&& \langle \Psi, H_N \Psi \rangle \nonumber \\
&\geq& \epsilon \langle \Psi, H_N \Psi \rangle + (1 - \epsilon) \sum_{j=1}^N \int_{\Lambda^N} |W_j(\x)|^2 |\nabla_j \Phi_j(\x)|^2 d\x \label{result after decomposition} \\
&& + (1 - \epsilon) \inf_u (\frac{\ell_1}{\ell_1 - 4\ell_0})^3 \sum_{i \neq j}^N \int_{\Lambda^{N-2}} dx_1 dx_2 \cdots \widehat{dx_i} \cdots \widehat{dx_j} \cdots dx_N \big(\sum_{\lambda} \int_{\Lambda_{u \lambda}^2} \widetilde{q_{ij}} |\Psi(\x)|^2 dx_i dx_j \big). \nonumber
\end{eqnarray}

The next step is to divide the kinetic energy into the small cubic cells. Let $P$ be the set of functions that assign particles to the small cubic cells, i.e.,
\begin{equation}
P = \{ p \; | \; p: \{1, 2, \cdots, N \} \to G \}.
\end{equation}
Let $p^{-1}(\lambda) = \{ i : p(i) = \lambda, 1 \leq i \leq N \}$. Then,
\begin{eqnarray}
&& \sum_{j=1}^N \int_{\Lambda^N} |W_j(\x)|^2 |\nabla_j \Phi_j(\x)|^2 d\x \nonumber \\
&=& \sum_{j=1}^N \sum_{p \in P} \int_{\Lambda^N} \Big( \prod_{k=1}^N 1(x_k \in \Lambda_{u p(k)}) dx_k \Big) |W_j(\x)|^2 |\nabla_j \Phi_j(\x)|^2 \label{division of kinetic energy} \\
&=& \sum_{p \in P} \sum_{\lambda} \int_{\Lambda^{N - |p^{-1}(\lambda)|}} \Big( \prod_{k: p(k) \neq \lambda} 1(x_k \in \Lambda_{u p(k)}) dx_k \Big) \nonumber \\
&& \hskip50pt \times \Big( \int_{\Lambda_{u \lambda}^{|p^{-1}(\lambda)|}} \sum_{j: p(j) = \lambda} |W_j(\x)|^2 |\nabla_j \Phi_j(\x)|^2 \prod_{k: p(k) = \lambda} dx_k \Big). \nonumber
\end{eqnarray}
Here, the last integral of the right hand side of \eqref{division of kinetic energy} represents the kinetic energy in $\Lambda_{u \lambda}$. Likewise, we also have
\begin{eqnarray}
&& \sum_{i \neq j}^N \int_{\Lambda^{N-2}} dx_1 dx_2 \cdots \widehat{dx_i} \cdots \widehat{dx_j} \cdots dx_N \big(\sum_{\lambda} \int_{\Lambda_{u \lambda}^2} \widetilde{q_{ij}} |\Psi(\x)|^2 dx_i dx_j \big) \nonumber \\
&=& \sum_{p \in P} \sum_{\lambda} \int_{\Lambda^{N - |p^{-1}(\lambda)|}} \Big( \prod_{k: p(k) \neq \lambda} 1(x_k \in \Lambda_{u p(k)}) dx_k \Big) \label{division of potential energy} \\
&& \hskip50pt \times \Big( \int_{\Lambda_{u \lambda}^{|p^{-1}(\lambda)|}} \sum_{i, j: i \neq j, p(i) = p(j) = \lambda} \widetilde{q_{ij}} |\Psi(\x)|^2 \prod_{k: p(k) = \lambda} dx_k \Big) \nonumber
\end{eqnarray}

For a given box $B$, Let $\N(B)$ be the function that indicates the number of particles in $B$, i.e.,
\begin{equation}
\N(B) := \sum_{i=1}^N 1(x_i \in B).
\end{equation}
Then, applying Lemma \ref{small cell estimate}, from \eqref{result after decomposition} \eqref{division of kinetic energy}, and \eqref{division of potential energy} we get
\begin{eqnarray}
&& \epsilon \langle \Psi, H_N \Psi \rangle + \frac{1}{2} \sum_{j=1}^N \int_{\Lambda^N} |W_j(\x)|^2 |\nabla_j \Phi_j(\x)|^2 d\x \nonumber \\
&& + (1 - \epsilon) (\frac{\ell_1}{\ell_1 - 4\ell_0})^3 \sum_{i \neq j}^N \int_{\Lambda^{N-2}} dx_1 dx_2 \cdots \widehat{dx_i} \cdots \widehat{dx_j} \cdots dx_N \big( \sum_{\lambda} \int_{\Lambda_{u \lambda}^2} \widetilde{q_{ij}} |\Psi(\x)|^2 dx_i dx_j \Big) \nonumber \\
&\geq& (1 - C \epsilon) \frac{4 \pi a}{\ell_1^3} \int_{\Lambda^N} \sum_{\lambda} f(\N(\Lambda_{u \lambda})) |\Psi(\x)|^2 d\x, \label{result after small cell estimate}
\end{eqnarray}
where $f$ is defined as in \eqref{definition of f}. Note that in \eqref{result after small cell estimate} we only used half of kinetic energy. Combining \eqref{result after energy estimate}, \eqref{result after decomposition}, and \eqref{result after small cell estimate}, we get
\begin{eqnarray}
\langle \Psi, H_N \Psi \rangle \geq (\frac{1}{2} - \epsilon) \sum_{j=1}^N \int_{\Lambda^N} |W_j(\x)|^2 |\nabla_j \Phi_j(\x)|^2 d\x + (1 - C \epsilon) \frac{4 \pi a}{\ell_1^3} \int_{\Lambda^N} \sum_{\lambda} f(\N(\Lambda_{u \lambda})) |\Psi(\x)|^2 d\x.
\end{eqnarray}

Now we consider larger cubic boxes of side length $\ell_2$. We let $\Lambda'_{u \theta}$ be a cubic box of side length $\ell_2$, which is a union of $2^{3h}$ $\Lambda_{u \lambda}$'s, where $\theta$ is an index for those larger cubic boxes. With a shorthand notation $\x = (x_1, x_2, \cdots, x_N)$, we obtain from Lemma \ref{large cell estimate} that
\begin{eqnarray}
&& (\frac{1}{2} - \epsilon) \sum_{j=1}^N \int_{\Lambda^N} |W_j(\x)|^2 |\nabla_j \Phi_j(\x)|^2 d\x + (1 - C \epsilon) \frac{4 \pi a}{\ell_1^3} \int_{\Lambda^N} \sum_{\lambda} f(\N(\Lambda_{u \lambda})) |\Psi(\x)|^2 d\x \nonumber \\
&\geq& (1 - C \epsilon) \frac{4 \pi a}{\ell_2^3} \int_{\Lambda^N} \sum_{\theta} \widetilde{f}(\N(\Lambda'_{u \theta})) |\Psi(\x)|^2 d\x,
\end{eqnarray}
where $\widetilde{f}$ is a convex function such that $\widetilde{f}(n) = n(n-1)$ if $n \leq \rho \ell_2^3$. So far we have proved that
\begin{eqnarray}
\langle \Psi, H_N \Psi \rangle \geq \inf_u \; (1 - C \epsilon) \frac{4 \pi a}{\ell_2^3} \int_{\Lambda^N} \sum_{\theta} \widetilde{f}(\N(\Lambda'_{u \theta})) |\Psi(\x)|^2 d\x.
\end{eqnarray}

We are left to minimize $\sum_{\theta} \widetilde{f}(\N(\Lambda'_{u \theta}))$. Since $\widetilde{f}$ is convex and $\sum_{\theta} \N(\Lambda'_{u \theta}) = N$, we can use Jensen's inequality to get, for any $x_1, x_2, \cdots, x_N$,
\begin{equation} \label{Jensen inequality}
\sum_{\theta} \widetilde{f}(\N(\Lambda'_{u \theta})) \geq \frac{L^3}{\ell_2^3} \widetilde{f}(\rho \ell_2^3) \geq \frac{L^3}{\ell_2^3} \rho \ell_2^3 (\rho \ell_2^3 - 1) \geq (1 - C \rho^{\frac{1}{3}} |\log \rho|^2) N \rho \ell_2^3.
\end{equation}
Note that this lower bound is independent of $u$. Hence
\begin{eqnarray}
\frac{\langle \Psi, H_N \Psi \rangle}{\langle \Psi, \Psi \rangle} \geq (1 - C \epsilon) \frac{4 \pi a}{\ell_2^3} N \rho \ell_2^3 = 4 \pi a \rho N (1 - C \epsilon).
\end{eqnarray}
This shows that, when $\rho$ is small enough, there exists a constant $C_0$ such that
\begin{equation}
\mathop{\lim_{N, L \to \infty}}_{N/L^3 = \rho} \frac{E(N, \Lambda)}{N} \geq 4 \pi a \rho ( 1 - C_0 \epsilon ),
\end{equation}
which was to be proved.

\end{proof}

\section{Lower bound estimates} \label{key lemmas}

In this section, we prove Lemma \ref{energy estimate}, Lemma \ref{decomposition}, and Lemma \ref{small cell estimate}, which were used in the proof of main theorem.

\subsection{Conversion into a soft potential}

\begin{proof} [Proof of Lemma \ref{energy estimate}]
Fix $j$ and consider $x_1, x_2, \cdots, \widehat{x_j}, \cdots, x_N$ to be fixed. From that $\Phi_j = W_j^{-1} \Psi$, we get
\begin{eqnarray} \label{ground state identity}
&& \int_{\Lambda} |\nabla_j \Psi(\x)|^2 dx_j + \frac{1}{2} \sum_{i: i \neq j} \int_{\Lambda} V(x_i - x_j) |\Psi(\x)|^2 dx_j \\
&=& \int_{\Lambda} |W_j(\x)|^2 |\nabla_j \Phi_j(\x)|^2 dx_j + \sum_{i: i \neq j} \int_{\Lambda} \Big[ W_j(\x) \big( -\Delta_j + \frac{1}{2} V(x_i - x_j) \big) W_j(\x) \Big] |\Phi_j(\x)|^2 dx_j \nonumber 
\end{eqnarray}

For each $i \neq j$, we have either $F_{ij} = 1$ or $F_{ij}=0$. Note that $F_{ij}$ is independent of $x_j$.

\begin{enumerate}

\item When $F_{ij} = 1$, consider $B(x_i, \ell_0)$, a ball of radius $\ell_0$ centered at $x_i$. When $x_j \in B(x_i, \ell_0)$, $W_j(\x) = 1 - \tau(\ell_0, x_i-x_j)$, thus
\begin{equation}
\big( -\Delta_j + \frac{1}{2} V(x_i - x_j) \big) W_j(\x) = e_0 (\ell_0) W_j(\x).
\end{equation}
When $x_j \notin B(x_i, \ell_0)$, $\big( -\Delta_j + \frac{1}{2} V(x_i - x_j) \big) W_j(\x) =0 $. Thus,
\begin{eqnarray}
&& \int_{\Lambda} \Big[ W_j(\x) \big( -\Delta_j + \frac{1}{2} V(x_i - x_j) \big) W_j(\x) \Big] |\Phi_j(\x)|^2 dx_j \nonumber \\
&=& \int_{\Lambda} e_0 (\ell_0) 1(|x_i - x_j| \leq \ell_0) |W_j(\x)|^2 |\Phi_j(\x)|^2 dx_j = \int_{\Lambda} F_{ij}(\x) q(\ell_0, x_i - x_j) |\Psi(\x)|^2 dx_j.
\end{eqnarray}

\item When $F_{ij} = 0, G_{ij} = 1$, consider $B(x_i, t_{ij})$, a ball of radius $t_{ij}$ centered at $x_i$. When $x_j \in B(x_i, t_{ij})$, $W_i(\x) = 1 - \tau(t_{ij}, x_i-x_j)$, thus
\begin{equation}
\big( -\Delta_j + \frac{1}{2} V(x_i - x_j) \big) W_j(\x) = e_0 (t_{ij}) W_j(\x).
\end{equation}
When $x_j \notin B(x_i, t_{ij})$, $W_j(\x) \big( -\Delta_j + \frac{1}{2} V(x_i - x_j) \big) W_j(\x) =0 $. Thus,
\begin{eqnarray}
&& \int_{\Lambda} \Big[ W_j(\x) \big( -\Delta_j + \frac{1}{2} V(x_i - x_j) \big) W_j(\x) \Big] |\Phi_j(\x)|^2 dx_j \\
&=& \int_{\Lambda} e_0 (t_{ij}) 1(|x_i - x_j| \leq t_{ij}) |W_j(\x)|^2 |\Phi_j(\x)|^2 dx_j = \int_{\Lambda} (1- F_{ij}(\x)) G_{ij} (\x) q(t_{ij}, x_i - x_j) |\Psi(\x)|^2 dx_j. \nonumber
\end{eqnarray}

\item When $F_{ij} = G_{ij} = 0$, we have that $W_j(\x) = 1$, thus we get
\begin{equation}
\int_{\Lambda} \Big[ W_j(\x) \big( -\Delta_j + \frac{1}{2} V(x_i - x_j) \big) W_j(\x) \Big] |\Phi_j(\x)|^2 dx_j \geq 0.
\end{equation}

\end{enumerate}

From cases (1)-(3), \eqref{ground state identity} implies
\begin{eqnarray}
&& \int_{\Lambda} \big( |\nabla_j \Psi(x_j)|^2 + \frac{1}{2} \sum_{i: i \neq j} V(x_i - x_j) |\Psi(x_j)|^2 \big) dx_j \nonumber \\
&\geq& \int_{\Lambda} |W_j(\x)|^2 |\nabla_j \Phi_j(\x)|^2 dx_j \label{one particle estimate} \\
&& + \sum_{i: i \neq j} \int_{\Lambda} \big[ F_{ij}(\x) q(\ell_0, x_i - x_j) + (1- F_{ij}(\x)) G_{ij} (\x) q(t_{ij}, x_i - x_j) \big] |\Psi(\x)|^2 dx_j. \nonumber
\end{eqnarray}

To get back to the $N$ particle problem, we first integrate \eqref{one particle estimate} and summing it over $j$ gives the desired lemma.
\end{proof}

\subsection{Decomposition of $\Lambda$}

\begin{proof} [Proof of Lemma \ref{decomposition}]
Recall that we let $u \in [-\ell_1/2, \ell_1/2)^3 = \Gamma$, $G = \ell_1 \mathbb{Z}^3 \cap \Lambda$, and for $\lambda \in G$, $\Lambda_{u \lambda}$ be a cubic box of side length $\ell_1$ centered at $(u + \lambda)$. Here, $u$ corresponds to the origin of the grid that divides $\Lambda$ into small boxes $\Lambda_{u \lambda}$. Note that the positions of those boxes depend on $u$.

Define $\widetilde{\chi}_{u \lambda}$ by
\begin{equation}
\widetilde{\chi}(x) :=
	\begin{cases}
	1 & \text{ if } x \in [\displaystyle -\frac{\ell_1}{2} + 2\ell_0, \frac{\ell_1}{2} - 2\ell_0)^3 \\
	0 & \text{ otherwise }
	\end{cases}.
\end{equation}
and $\widetilde{\chi}_{u \lambda} := \widetilde{\chi} (x-u-\lambda)$. If $\widetilde{\chi}_{u \lambda}(x_i) = 1$, then $x_i \in \Lambda_{u \lambda}$ and $x_i$ is not within distance $2 \ell_0$ from the boundary of $\Lambda_{u \lambda}$. Note that
\begin{equation}
\frac{1}{|\Gamma|} \int_{\Gamma} du \sum_{\lambda \in G} \widetilde{\chi}_{u \lambda}(x) = (\frac{\ell_1 - 4\ell_0}{\ell_1})^3
\end{equation}
for any $x \in \Lambda$. This means that the probability of having $x \in \Lambda$ to satisfy $\widetilde{\chi}_{u \lambda} = 1$ for a $\lambda$ is $(\ell_1 - 4 \ell_0)^3 / \ell_1^3$.

From the definitions of $\widetilde{q_{ij}}$ and $\widetilde{\chi}_{u \lambda}$, we have
\begin{equation}
\widetilde{q_{ij}}(\x) = \sum_{\lambda \in G} \widetilde{\chi}_{u \lambda}(x_i) q_{ij}(\x).
\end{equation}
($\widetilde{q_{ij}}$ depends on $u$, but we omit it.) Thus, for any $\Psi(\x) \in L^2 (\Lambda^N)$,
\begin{eqnarray}
&& \frac{1}{|\Gamma|} \int_{\Gamma} du \int_{\Lambda^2} \widetilde{q_{ij}} |\Psi(\x)|^2 dx_i dx_j = \frac{1}{|\Gamma|} \int_{\Lambda^2} \int_{\Gamma} du \sum_{\lambda \in G} \widetilde{\chi}_{u \lambda}(x_i) q_{ij} |\Psi(\x)|^2 dx_i dx_j \nonumber \\
&=& (\frac{\ell_1 - 4\ell_0}{\ell_1})^3 \int_{\Lambda^2} q_{ij} |\Psi(\x)|^2 dx_i dx_j. \label{q' expansion}
\end{eqnarray}
Integrating \eqref{q' expansion} with respect to $dx_1 dx_2 \cdots \widehat{dx_i} \cdots \widehat{dx_j} \cdots dx_N$ and summing over $i, j$ gives
\begin{eqnarray}
&& \sum_{i \neq j}^N \int_{\Lambda^N} q_{ij} |\Psi(\x)|^2 d\x = \frac{1}{|\Gamma|} \int_{\Gamma} du (\frac{\ell_1}{\ell_1 - 4\ell_0})^3 \sum_{i \neq j}^N \int_{\Lambda^N} \widetilde{q_{ij}} |\Psi(\x)|^2 d\x \nonumber \\
&\geq& \inf_u (\frac{\ell_1}{\ell_1 - 4\ell_0})^3 \sum_{i \neq j}^N \int_{\Lambda^N} \widetilde{q_{ij}} |\Psi(\x)|^2 d\x.
\end{eqnarray}
This proves the lemma.
\end{proof}

\subsection{Lower bound estimate - Small cubic cell}

\begin{proof} [Proof of Lemma \ref{small cell estimate}]
We consider the following cases:

\begin{enumerate}

\item When $n \leq 9 \rho \ell_1^3$:

When $n \leq 1$, $f(n)=0$ and the lemma is trivial. Suppose that $n \geq 2$. Define an operator
\begin{equation}
T_j = - W_j^{-1} \nabla_j W_j^2 \nabla_j W_j^{-1}
\end{equation}
on the functions in $L^2 (B^n)$ with $W_j$-Neumann boundary conditions in B in the sense that, for a function $\psi \in L^2 (B^n)$,
\begin{equation}
T_j \psi = - W_j^{-1} \nabla_j (W_j^2 \nabla_j (W_j^{-1} \psi)).
\end{equation}
For $i, j \in \{1, 2, \cdots, n \}$, $i \neq j$, we want to consider an operator
\begin{equation}
\frac{T_i}{4(n-1)} + \frac{T_j}{4(n-1)} + (\frac{\ell_1}{\ell_1 - 4\ell_0})^3 \widetilde{q_{ij}},
\end{equation}
which is defined on functions in $L^2 (B^n)$ satisfying $W_i$-Neumann boundary conditions and $W_j$-Neumann boundary conditions.

We first estimate a lower bound for 
\begin{equation}
\frac{T_j}{4(n-1)} + (\frac{\ell_1}{\ell_1 - 4\ell_0})^3 \widetilde{q_{ij}}
\end{equation}
by Temple's inequality \cite{T}, with $[4(n-1)]^{-1} T_j$ as the unperturbed part in the first order perturbation theory.

To find the gap of $T_j$, we first notice that $W_j$, which is defined in \eqref{definition of W}, is the ground state of $T_j$. For a function $\psi(x_j)$ with $W_j$-Neumann boundary condition in $B$ with
\begin{equation}
\int_B W_j \psi(x_j) dx_j = 0,
\end{equation}
we apply Poincare's inequality to obtain
\begin{equation}
\int_B | \psi(x_j) - \bar{\psi} |^2 dx_j \leq C |B|^{\frac{2}{3}} \int_B |\nabla \psi(x_j)|^2 dx_j,
\end{equation}
where
\begin{equation}
\bar{\psi} = \frac{1}{|B|} \int_B \psi(x_j) dx_j.
\end{equation}
We also have
\begin{eqnarray}
&& \int_B | \psi(x_j) - \bar{\psi} |^2 dx_j = \int_B | \psi(x_j)|^2 dx_j - \frac{1}{|B|} \big( \int_B \psi(x_j) dx_j \big)^2 \nonumber \\
&=& \int_B | \psi(x_j)|^2 dx_j - \frac{1}{|B|} \big( \int_B (1- W_j) \psi(x_j) dx_j \big)^2 \geq \int_B | \psi(x_j)|^2 dx_j - \frac{c_0^2}{|B|} \big( \int_B |\psi(x_j)| dx_j \big)^2 \nonumber \\
&\geq& (1 - c_0^2) \int_B | \psi(x_j)|^2 dx_j.
\end{eqnarray}
Thus, we can see that, in a box of side length $\ell_1$,
\begin{equation}
(\text{gap of } T_j ) \geq C \ell_1^{-2}.
\end{equation}

In order to use Temple's inequality, we first need to check that
\begin{equation}
\frac{T_j}{4(n-1)} + (\frac{\ell_1}{\ell_1 - 4\ell_0})^3 \widetilde{q_{ij}} \geq 0,
\end{equation}
which is obvious since $T_j \geq 0$ and $\widetilde{q_{ij}} \geq 0$. Let $\langle F \rangle_{W_i}$ denotes
\begin{equation}
\langle F \rangle_{W_i} = \int_{B} dx_i F |W_i|^2 / \int_{B} dx_i |W_i|^2.
\end{equation}
We also need to have
\begin{equation}
(\text{gap of } \frac{T_j}{4(n-1)} ) \gg (\frac{\ell_1}{\ell_1 - 4\ell_0})^3 \langle \widetilde{q_{ij}} \rangle_{W_j},
\end{equation}
which can be easily seen from that
\begin{eqnarray}
&& (\text{gap of } \frac{T_j}{4(n-1)} ) \geq C n^{-1} \ell_1^{-2} \geq C \rho^{\frac{2}{3}} |\log \rho|^{-\frac{5}{3} - 5 \eta} \nonumber \\
&\gg& C \rho^{\frac{2}{3}} |\log \rho|^{-2} \geq C \| q_{ij} \|_{\infty} \geq (\frac{\ell_1}{\ell_1 - 4\ell_0})^3 \langle \widetilde{q_{ij}} \rangle_{W_j}.
\end{eqnarray}
Thus, we can indeed use Temple's inequality to obtain that
\begin{equation} \label{perturbation series}
\frac{T_j}{4(n-1)} + (\frac{\ell_1}{\ell_1 - 4\ell_0})^3 \widetilde{q_{ij}} \geq (\frac{\ell_1}{\ell_1 - 4\ell_0})^3 \langle \widetilde{q_{ij}} \rangle_{W_j} - C \frac{\langle \widetilde{q_{ij}}^2 \rangle_{W_j} - \langle \widetilde{q_{ij}} \rangle_{W_j}^2}{n^{-1} \ell_1^{-2}}.
\end{equation}

We want to estimate $\langle \widetilde{q_{ij}} \rangle_{W_j}$ and $\langle \widetilde{q_{ij}}^2 \rangle_{W_j}$. It follows from Lemma \ref{neumann problem} that, for all $\ell_{-1} \leq \kappa \leq \ell_0$,
\begin{equation}
\tau(\kappa, x_i - x_j) \leq \frac{C}{|x_i - x_j|}
\end{equation}
Hence,
\begin{equation}
\widetilde{q_{ij}} |W_j|^2 \geq \widetilde{q_{ij}} (1 - \frac{C}{|x_i - x_j|})^2 \geq \widetilde{q_{ij}} (1 - \frac{C}{|x_i - x_j|}).
\end{equation}

Let $S_{-1}$ and $S_0$ be sets of all points in $B$ that are not within distance $2\ell_{-1}$ and $2\ell_0$, respectively, from any of $x_1, x_2, \cdots, \widehat{x_i}, \cdots, \widehat{x_j}, \cdots, x_N$, i.e.,
\begin{equation}
S_{-1} = \{ x \in B : \forall k \neq i, j, |x-x_k| > 2\ell_{-1} \},
\end{equation}
\begin{equation}
S_0 = \{ x \in B : \forall k \neq i, j, |x-x_k| > 2\ell_0 \}.
\end{equation}
Let $\widetilde{B}$ be a set of all points in $B$ that are not within distance $2 \ell_0$ to the boundary of $B$, i.e.,
\begin{equation}
\widetilde{B} = \{ x \in B : d(x, \partial B) \geq 2 \ell_0 \}.
\end{equation}
By definition, $\widetilde{q_{ij}} = 0$ if $x_i \in B \backslash S_{-1}$ or $x_i \in B \backslash \widetilde{B}$, and $\widetilde{q_{ij}} = q_{ij}$ if $x_i \in S_{-1} \cap \widetilde{B}$.

When $x_i \in S_0 \cap \widetilde{B}$,
\begin{eqnarray}
\int_{B} \widetilde{q_{ij}} |W_j|^2 dx_j &\geq& \int_{B} \widetilde{q_{ij}} (1 - \frac{C}{|x_i - x_j|}) dx_j \geq e_0 (\ell_0) \int_{|x_i - x_j| \leq \ell_0} (1 - \frac{C}{|x_i - x_j|}) dx_j \nonumber \\
&\geq& (1 - \frac{C}{\ell_0}) e_0 (\ell_0) \int_{|x_i - x_j| \leq \ell_0} 1 dx_j \geq 4 \pi a (1 - \frac{C}{\ell_0}). \label{q w estimate 1}
\end{eqnarray}
When $x_i \in (S_{-1} \backslash S_0) \cap \widetilde{B}$,
\begin{eqnarray}
\int_{B} \widetilde{q_{ij}} |W_j|^2 dx_j &\geq& \int_{B} \widetilde{q_{ij}} (1 - \frac{C}{|x_i - x_j|}) dx_j \geq e_0 (t_{ij}) \int_{|x_i - x_j| \leq t_{ij}} (1 - \frac{C}{|x_i - x_j|}) dx_j \nonumber \\
&\geq& (1 - \frac{C}{t_{ij}}) e_0 (t_{ij}) \int_{|x_i - x_j| \leq t_{ij}} 1 dx_j \geq 4 \pi a (1 - \frac{C}{t_{ij}}). \label{q w estimate 2}
\end{eqnarray}

For $\langle \widetilde{q_{ij}}^2 \rangle_{W_j}$, when $x_i \in S_0 \cap \widetilde{B}$,
\begin{eqnarray} \label{q w estimate 3}
\int_{B} \widetilde{q_{ij}}^2 |W_j|^2 dx_j \leq \int_{B} \widetilde{q_{ij}}^2 dx_j \leq C \ell_0^{-3}, 
\end{eqnarray}
and, when $x_i \in (S_{-1} \backslash S_0) \cap \widetilde{B}$,
\begin{eqnarray} \label{q w estimate 4}
\int_{B} \widetilde{q_{ij}}^2 |W_j|^2 dx_j \leq \int_{B} \widetilde{q_{ij}}^2 dx_j \leq C t_{ij}^{-3}.
\end{eqnarray}

Since we know from Lemma \ref{properties of W} that
\begin{equation} \label{W bound}
\ell_1^3 (1 - C n \frac{\ell_0^2}{\ell_1^3} - C \ell_1^{-1}) \leq \int_B |W_j|^2 dx_j \leq \ell_1^3,
\end{equation}
from \eqref{q w estimate 1} and \eqref{q w estimate 2}, we get
\begin{eqnarray}
&& (\frac{\ell_1}{\ell_1 - 4\ell_0})^3 \langle \widetilde{q_{ij}} \rangle_{W_j} = (\frac{\ell_1}{\ell_1 - 4\ell_0})^3 \int_{B} \widetilde{q_{ij}} |W_j|^2 dx_j \big/ \int_B |W_j|^2 dx_j \nonumber \\
&\geq& (\frac{\ell_1}{\ell_1 - 4\ell_0})^3 \frac{4 \pi a}{\ell_1^3} \left( (1 - \frac{C}{\ell_0}) \cdot 1(x_i \in S_0 \cap \widetilde{B}) + (1 - \frac{C}{t_{ij}}) \cdot 1(x_i \in (S_{-1} \backslash S_0) \cap \widetilde{B}) \right), \label{perturbation estimate 1}
\end{eqnarray}
and, from \eqref{q w estimate 3} and \eqref{q w estimate 4}, we get
\begin{eqnarray}
&& \frac{\langle \widetilde{q_{ij}}^2 \rangle_{W_j} - \langle \widetilde{q_{ij}} \rangle_{W_j}^2}{n^{-1} \ell_1^{-2}} \leq \frac{\langle \widetilde{q_{ij}}^2 \rangle_{W_j}}{n^{-1} \ell_1^{-2}} = n \ell_1^2 \int_{B} \widetilde{q_{ij}}^2 |W_j|^2 dx_j \big/ \int_{B} |W_j|^2 dx_j \nonumber \\
&\leq& C \frac{n \ell_1^2}{\ell_0^3} \ell_1^{-3} \cdot 1(x_i \in S_0 \cap \widetilde{B}) + C \frac{n \ell_1^2}{t_{ij}^3} \ell_1^{-3} \cdot 1(x_i \in (S_{-1} \backslash S_0) \cap \widetilde{B}). \label{perturbation estimate 2}
\end{eqnarray}

Inserting \eqref{perturbation estimate 1} and \eqref{perturbation estimate 2} into \eqref{perturbation series}, we obtain that
\begin{eqnarray}
&& \frac{T_j}{4(n-1)} + (\frac{\ell_1}{\ell_1 - 4\ell_0})^3 \widetilde{q_{ij}} \nonumber \\
&\geq& (\frac{\ell_1}{\ell_1 - 4\ell_0})^3 \frac{4 \pi a}{\ell_1^3} \left( (1 - \frac{C}{\ell_0} - C \frac{n \ell_1^2}{\ell_0^3}) \cdot 1(x_i \in S_0 \cap \widetilde{B}) + (1 - \frac{C}{t_{ij}} - C \frac{n \ell_1^2}{t_{ij}^3}) \cdot 1(x_i \in (S_{-1} \backslash S_0) \cap \widetilde{B}) \right) \nonumber \\
&\geq& (\frac{\ell_1}{\ell_1 - 4\ell_0})^3 \frac{4 \pi a}{\ell_1^3} \left( (1 - C \frac{n \ell_1^2}{\ell_0^3}) \cdot 1(x_i \in S_0 \cap \widetilde{B}) + (1 - C \frac{n \ell_1^2}{t_{ij}^3}) \cdot 1(x_i \in (S_{-1} \backslash S_0) \cap \widetilde{B}) \right). \label{perturbation 1}
\end{eqnarray}
Here, for the last inequality, we used $n \ell_1^2 \geq \ell_1^2 \gg \ell_0^2$ and, when $x_i \in (S_{-1} \backslash S_0) \cap \widetilde{B}$, $n \ell_1^2 \geq \ell_1^2 \gg \ell_0^2 \geq t_{ij}^2$.

Let
\begin{equation}
\xi_i (\x_n) = \frac{4 \pi a}{\ell_1^3} \left( (1 - C \frac{n \ell_1^2}{\ell_0^3}) \cdot 1(x_i \in S_0 \cap \widetilde{B}) + (1 - C \frac{n \ell_1^2}{t_{ij}^3}) \cdot 1(x_i \in (S_{-1} \backslash S_0) \cap \widetilde{B}) \right).
\end{equation}
Note that $S_0$ and $S_{-1}$ are independent of $x_i$ and $x_j$, and $\xi_i$ is independent of $x_j$. We apply Temple's inequality to 
\begin{equation}
\frac{T_i}{4(n-1)} + (\frac{\ell_1}{\ell_1 - 4\ell_0})^3 \xi_i
\end{equation}
with $[4(n-1)]^{-1} T_i$ as the unperturbed part. Then, we get
\begin{eqnarray}
\frac{T_i}{4(n-1)} + (\frac{\ell_1}{\ell_1 - 4\ell_0})^3 \xi_i \geq (\frac{\ell_1}{\ell_1 - 4\ell_0})^3 \langle \xi_i \rangle_{W_i} - C \frac{\langle \xi_i^2 \rangle_{W_i}}{n^{-1} \ell_1^{-2}}. \label{perturbation 2}
\end{eqnarray}

We now estimate $\langle \xi_i \rangle_{W_i}$. By definition,
\begin{eqnarray}
&& \langle \xi_i \rangle_{W_i} \geq \ell_1^{-3} \int_{S_{-1} \cap \widetilde{B}} \xi_i |W_i|^2 dx_i \nonumber \\
&\geq& \frac{4 \pi a}{\ell_1^3} (1 - C \frac{n \ell_1^2}{\ell_0^3}) \int_{S_{-1} \cap \widetilde{B}} \ell_1^{-3} |W_i|^2 dx_i - \frac{C}{\ell_1^3} \int_{(S_{-1} \backslash S_0) \cap \widetilde{B}} \ell_1^{-3} \frac{n \ell_1^2}{t_{ij}^3} dx_i. \label{xi estimate 1}
\end{eqnarray}
To estimate the last term in the right hand side, we note that (1) $t_{ij} = \frac{1}{2} |x_i - x_k|$ for some $x_k$ other than $x_j$ and (2) the union of annuli $\{ x: 2\ell_{-1} \leq |x - x_k| \leq 2\ell_0 \}$ for all $x_k \in B$ other than $x_i$ and $x_j$ covers $(S_{-1} \backslash S_0)$. Thus,
\begin{equation}
\int_{(S_{-1} \backslash S_0) \cap \widetilde{B}} t_{ij}^{-3} dx_i \leq \sum_{k: k \neq i} \int_{2\ell_{-1} \leq |x_i - x_k| \leq 2\ell_0} |x_i - x_k|^{-3} dx_i. \label{xi estimate 2}
\end{equation}
To estimate the first term in the right hand side of \eqref{xi estimate 1}, we use
\begin{eqnarray}
\int_{S_{-1} \cap \widetilde{B}} |W_i|^2 dx_i \geq \int_{S_{-1} \cap \widetilde{B}} dx_i - \int_B (1 - |W_i|^2) dx_i \geq |S_{-1} \cap \widetilde{B}| - C n \ell_0^2, \label{xi estimate 3}
\end{eqnarray}
where the last inequality follows from Lemma \ref{properties of W}. Since $|S_{-1}| \geq \ell_1^3 - C n (\ell_{-1})^3$ and $|\widetilde{B}| = (\ell_1 - 4 \ell_0)^3$, from \eqref{xi estimate 1}, \eqref{xi estimate 2}, and \eqref{xi estimate 3}, we get
\begin{eqnarray}
&& \langle \xi_i \rangle_{W_i} \nonumber \\
&\geq& \frac{|S_{-1} \cap \widetilde{B}|}{\ell_1^3} \frac{4 \pi a}{\ell_1^3} (1 - C \frac{n \ell_1^2}{\ell_0^3}) (1 - C n \frac{\ell_0^2}{\ell_1^3}) - C n \sum_{k: k \neq i} \int_{2\ell_{-1} \leq |x_i - x_k| \leq 2\ell_0} \ell_1^{-4} |x_i - x_k|^{-3} dx_i \nonumber \\
&\geq& (\frac{\ell_1 - 4\ell_0}{\ell_1})^3 \frac{4 \pi a}{\ell_1^3} (1 - C n \frac{(\ell_{-1})^3}{\ell_1^3} - C \frac{n \ell_1^2}{\ell_0^3} - C n \frac{\ell_0^2}{\ell_1^3} - C n^2 \frac{|\log \rho|}{\ell_1}). \label{perturbation 3}
\end{eqnarray}

Since $\langle \xi_i^2 \rangle_{W_i} \leq \| \xi_i \|_{\infty}^2 \leq C \ell_1^{-6}$,
\begin{equation} \label{perturbation 4}
\frac{\langle \xi_i^2 \rangle_{W_i}}{n^{-1} \ell_1^{-2}} \leq C n \ell_1^{-4}.
\end{equation}
Thus, combining \eqref{perturbation 1}, \eqref{perturbation 2}, \eqref{perturbation 3}, and \eqref{perturbation 4}, we obtain that
\begin{eqnarray}
&& \frac{T_i}{4(n-1)} + \frac{T_j}{4(n-1)} + (\frac{\ell_1}{\ell_1 - 4\ell_0})^3 \widetilde{q_{ij}} \nonumber \\
&\geq& (1 - C n \frac{(\ell_{-1})^3}{\ell_1^3} - C \frac{n \ell_1^2}{\ell_0^3} - C n \frac{\ell_0^2}{\ell_1^3} - C n^2 \frac{|\log \rho|}{\ell_1}) \frac{4 \pi a}{\ell_1^3} \\
&\geq& (1 - C \rho^{\frac{1}{3}} |\log \rho|^3) \frac{4 \pi a}{\ell_1^3}. \nonumber
\end{eqnarray}
This implies, when $x_{n+1}, \cdots, x_N$ are outside of $B$, for any $\psi(\x_n) \in L^2 (B^n)$ with $\phi_l = W_l^{-1} \psi$ for $l=1, 2, \cdots, n$,
\begin{eqnarray}
&& \frac{1}{4(n-1)} \int_{B^n} \big( |W_i|^2 |\nabla_i \phi_i(\x_n)|^2 + |W_j|^2 |\nabla_j \phi_j(\x_n)|^2 \big) d\x_n + (\frac{\ell_1}{\ell_1 - 4\ell_0})^3 \int_{B^n} \widetilde{q_{ij}} |\psi(\x_n)|^2 d\x_n \nonumber \\
&\geq& (1 - C \epsilon) \frac{4 \pi a}{\ell_1^3} \| \psi \|_2^2.
\end{eqnarray}

This shows how we can get the lower bound for the ground state energy in a soft potential regime. $\widetilde{q_{ij}}$ depends on particles other than $x_i$ and $x_j$, but overall effect from them is insignificant and becomes a small error.

We apply this argument to all $1 \leq i, j \leq n$, $i \neq j$. After summing over all indices $i$ and $j$, we get
\begin{eqnarray}
&& \frac{1}{2} \sum_{j=1}^n \int_{B^n} |W_j|^2 |\nabla_j \phi_j(\x_n)|^2 d\x_n + (\frac{\ell_1}{\ell_1 - 4\ell_0})^3 \sum_{i \neq j}^{n} \int_{B^n} \widetilde{q_{ij}} |\psi(\x_n)|^2 d\x_n \nonumber \\
&\geq& n(n-1) (1 - C \epsilon) \frac{4 \pi a}{\ell_1^3} \| \psi \|_2^2.
\end{eqnarray}

\item When $9 \rho \ell_1^3 < n \leq 9 \epsilon^{-1} \rho \ell_1^3$:

Let $p := 9 \rho \ell_1^3$. Here, $n$ satisfies that
\begin{equation} \label{requirement for n}
n \frac{(\ell_{-1})^3}{\ell_1^3} \ll 1, \;\;\; n p \frac{|\log \rho|}{\ell_1} \ll 1, n \frac{\ell_0^2}{\ell_1^3} \ll 1.
\end{equation}

We form particle groups in $B$, each of which contains at most $p$ particles, such that $G_1 = \{ x_1, x_2, \cdots, x_p \} $, $G_2 = \{ x_{p+1}, \cdots, x_{2p} \} $, $\cdots$, $G_{\lfloor n/p \rfloor} = \{ x_{(\lfloor n/p \rfloor - 1) p + 1}, \cdots, x_{\lfloor n/p \rfloor p} \}$, $G_{\lfloor n/p \rfloor + 1} = \{ x_{\lfloor n/p \rfloor p + 1}, \cdots, x_n \}$. Note that the number of groups with $p$ particles is $\lfloor n/p \rfloor$ and $\lfloor n/p \rfloor \geq n/2p$.

For $i, j \in G_k$, we consider 
\begin{equation}
\frac{T_i}{4(p-1)} + \frac{T_j}{4(p-1)} + (\frac{\ell_1}{\ell_1 - 4\ell_0})^3 \widetilde{q_{ij}}
\end{equation}
which is defined on functions in $L^2 (B^n)$ satisfying $W_i$-Neumann boundary conditions and $W_j$-Neumann boundary conditions. We then use the Temple's inequality as in case (1) to get, for any $\psi(\x_n) \in L^2 (B^n)$ with $\phi_l = W_l^{-1} \psi$ for $l=1, 2, \cdots, n$,
\begin{eqnarray}
&& \frac{1}{4(p-1)} \int_{B^n} \big( |W_i|^2 |\nabla_i \phi_i(\x_n)|^2 + |W_j|^2 |\nabla_j \phi_j(\x_n)|^2 \big) d\x_n + (\frac{\ell_1}{\ell_1 - 4\ell_0})^3 \int_{B^n} \widetilde{q_{ij}} |\psi_j(\x_n)|^2 d\x_n \nonumber \\
&\geq& (1 - C n \frac{(\ell_{-1})^3}{\ell_1^3} - C p \frac{\ell_1^2}{\ell_0^3} - C n \frac{\ell_0^2}{\ell_1^3} - C n p \frac{|\log \rho|}{\ell_1}) \frac{4 \pi a}{\ell_1^3} \| \psi \|_2^2.
\end{eqnarray}

We apply this inequality to all particles in the particle group $G_1$, then to all particle groups in $B$. Using \eqref{requirement for n}, we obtain
\begin{eqnarray}
&& \frac{1}{2} \sum_{j=1}^n \int_{B^n} |W_j|^2 |\nabla_j \phi_j(\x_n)|^2 d\x_n + (\frac{\ell_1}{\ell_1 - 4\ell_0})^3 \sum_{i \neq j}^{n} \int_{B^n} \widetilde{q_{ij}} |\psi_j(\x_n)|^2 d\x_n \nonumber \\
&\geq& \lfloor \frac{n}{p} \rfloor p (p-1) (1 - C n \frac{(\ell_{-1})^3}{\ell_1^3} - C p \frac{\ell_1^2}{\ell_0^3} - C n \frac{\ell_0^2}{\ell_1^3} - C n p \frac{|\log \rho|}{\ell_1}) \frac{4 \pi a}{\ell_1^3} \| \psi \|_2^2 \\
&\geq& \frac{n}{2} (p-1) \big( \frac{8}{9} \big) \frac{4 \pi a}{\ell_1^3} \| \psi \|_2^2 \geq (4 \rho \ell_1^3 - 1) n (1 - C \epsilon) \frac{4 \pi a}{\ell_1^3} \| \psi \|_2^2 \nonumber.
\end{eqnarray}

We note that, while in Lemma \ref{small cell estimate} the two cases $n \leq 2 \rho \ell_1^3$ and $n > 2 \rho \ell_1^3$ are distinguished, we choose the size of $p$ to be $9 \rho \ell_1^3$. If we would choose the size of $p$ to be $2 \rho \ell_1^3$, we would get a lower bound that is not convex at $n = 2 \rho \ell_2^3$, since the lower bound would increase proportionally to $n$. We need a space to connect two different formulas so that the resulting lower bound becomes convex, thus the distinction is between $n \leq 9 \rho \ell_1^3$ and $n > 9 \rho \ell_1^3$.

\item When $n > 9 \epsilon^{-1} \rho \ell_1^3$:

In this case, we only use the term $\epsilon E(n, B)$ to prove the lemma and ignore the other terms, since they are non-negative. It is known that the ground state energy of $n$ particle system in a box of side length $\ell_1$,
\begin{equation} \label{a priori lower bound}
E(n, B) \geq 4 \pi a (\frac{n}{\ell_1^3}) n (1-C (\frac{n}{\ell_1^3})^{\frac{1}{17}})
\end{equation}
when the density $(n/\ell_1^3)$ is sufficiently small and 
\begin{equation}
\ell_1 \geq C \big( \frac{n}{\ell_1^3} \big)^{-\frac{6}{17}}.
\end{equation}
(See Theorem 2.4 in \cite{LSSY}.)

In our case, if the box $B$ contains $m := 9 \epsilon^{-1} \rho \ell_1^3$ particles, then
\begin{equation}
\ell_1 \geq C \rho^{-\frac{1}{3}} |\log \rho|^{\frac{1}{3}} \gg C \rho^{-\frac{4}{17}} |\log \rho|^{\frac{18}{17}} = C (\frac{m}{\ell_1^3})^{-\frac{6}{17}},
\end{equation}
and the density in this case
\begin{equation}
\frac{m}{\ell_1^3} = 9 \epsilon^{-1} \rho = 9 \rho^{\frac{2}{3}} |\log \rho|^{-3} \to 0
\end{equation}
as $\rho \to 0$. Thus we can indeed use \eqref{a priori lower bound} to obtain that
\begin{equation}
E(m, B) \geq 4 \pi a \frac{m^2}{\ell_1^3} (1-C (\frac{2m}{\ell_1^3})^{\frac{1}{17}})
\end{equation}

To find a lower bound of $E(n, B)$, we form particle groups in $B$, each of which contains $m$ particles. Since we have $\lfloor n/m \rfloor$ groups of size $m$, by superadditivity,
\begin{eqnarray}
E(n, B) \geq \lfloor \frac{n}{m} \rfloor E(m, B) \geq (4 \rho \ell_1^3) n (1 - C \epsilon) \frac{4 \pi a}{\ell_1^3}.
\end{eqnarray}

\end{enumerate}

Since, $n$ falls into one of the above categories, we get the desired lemma.
\end{proof}

\section{Box doubling method} \label{doubling a box}

In this section, we prove Lemma \ref{large cell estimate}.

\subsection{Lower bound estimate - Large cubic cell}

In order to show Lemma \ref{large cell estimate}, we need to prove a result analogous to Lemma \ref{small cell estimate} when $\Lambda_{\ell_2}$, a box of side length $\ell_2 $ is given. We note:
\begin{equation}
\ell_1\ll \ell_2\sim \rho^{-\frac{4}{9}} |\log \rho|^{-\frac{2}{3}}
\end{equation}
More specifically, for a box B with a side length between $\ell_1$ and $\ell_2$, we will show that the energy in $B$ with $n$ particles ($n\sim \rho |B|$) retains the form, 
\begin{equation} \label{potential form}
4\pi a (1-C\epsilon)|B|^{-1} n (n-1).
\end{equation}
We can see that the error factor $1/n$ that comes from the ratio between $n(n-1)$ and $n^2$ becomes smaller as $n$ increases, and eventually it becomes $1/n \sim \rho^{\frac{1}{3}} |\log \rho|^2\ll\epsilon$ when the side length of $B$ becomes $\ell_2$, i.e, $n \sim \rho \ell_2^3$.

To demonstrate how to enlarge the size of box while retaining the form \eqref{potential form}, we first consider a simple case where we have only two adjacent boxes $\Lambda_A$ and $\Lambda_B$ with the same size. Suppose that there are $n$ particles in $\Lambda_A \cup \Lambda_B$. Let $n_A$ and $n_B $ denote the number of the particles in $\Lambda_A$ and $\Lambda_B$. We assume the potential energy in $\Lambda_A$ and $\Lambda_B$ as $n_A (n_A - 1) / |\Lambda_A|$ and $n_B (n_B - 1) / |\Lambda_B|$, which has the form in \eqref{potential form}. For $\alpha\in \mathbb R$, we define the Hamiltonian as 
\begin{equation}
H_\alpha=\alpha\sum_{i=1}^n (-\Delta_i) + n_A (n_A - 1) / |\Lambda_A| + n_B (n_B - 1) / |\Lambda_B|
\end{equation}
When $\alpha=0$(no kinetic energy), the ground state energy of this Hamiltonian is 
$n (n - 2) / |\Lambda_A \cup \Lambda_B|$. But when $\alpha=\infty$(particles are uniformly distributed in $\Lambda_A \cup \Lambda_B$), the ground state energy is equal to $n (n - 1) / |\Lambda_A \cup \Lambda_B|$, 
which gives the desired form \eqref{potential form}. 

We will show that instead of $-\alpha\Delta_i$, a small potion of $T_i$ can also guarantee the almost-uniform distribution and the desired form \eqref{potential form}.

This heuristic argument shows our basic strategy in this section, which we call `box doubling method.' Recall that $\ell_2 = 2^h \ell_1$. In this method, we begin from the first step where we have $2^{3h}$ small cubic cells of side length $\ell_1 $. We consider $(2^{3h}/2)$ pairs of adjacent boxes, and for each pair that consists of two adjacent boxes of same size $\Lambda_A$ and $\Lambda_B$. As explained above, we can get a lower bound for the energy of $n$ particle system in $\Lambda_A \cup \Lambda_B$ at expense of small potion of $T_i$'s. In this way, we can effectively make the size of each box doubled, since the new `potential energy term' in $\Lambda_A \cup \Lambda_B$ also has the form in \eqref{potential form}(when the density in $\Lambda_A \cup \Lambda_B$ is about $\rho$).

At the end of the first step, or the beginning of the 2nd step, we have $2^{3h-1}$ boxes whose dimensions are $\ell_1 \times \ell_1 \times 2 \ell_1$. In the 2nd step, we consider $2^{3h-2}$ pairs of those boxes and perform the above process again for all the pairs.  Keep using this method. At the beginning of the $s$-th step, we have $2^{3h-s}$ boxes, and after applying the above method to $2^{3h-s-1}$ pairs of boxes, the number of boxes gets halved and the size of each box doubled. And the new `potential energy terms' in new boxes also have the form in \eqref{potential form}(when the density in new boxes is about $\rho$)

We keep repeating it until the side length of a box becomes $\ell_2$, which is when $s = 3h$, and we only have one box left. The form of the potential term, \eqref{potential form} remains the same throughout this procedure, and it can lead us to the desired result, Lemma \ref{large cell estimate}. We will make this argument rigorous in this section.

Before we begin the proof, we introduce definitions that will be used throughout this section.

\begin{itemize}

\item $s$ is a non-negative integer that satisfies $1 \leq s \leq 3h$, where $\ell_2 = 2^h \ell_1$. We let 
\begin{equation}
\ell(s) := 2^{\lfloor \frac{s-1}{3} \rfloor} \ell_1.
\end{equation}
$\ell(s)$ satisfies $\rho^{-\frac{1}{3}} |\log \rho|^{\frac{1}{3} + \eta} \sim \ell_1 \leq \ell(s) \leq \ell_2 \sim \rho^{-\frac{4}{9}} |\log \rho|^{-\frac{2}{3}}$, where $0 < \eta < 1/15$.

\end{itemize}

This $s$ is a label keeping track of which step we are at. We begin from $s=1$ and our method ends when $s=3h$.

\begin{itemize}

\item $\Lambda_A$ and $\Lambda_B$ are boxes such that the volume of each box $|\Lambda_A| = |\Lambda_B| = 2^{s-1} \ell_1^3$ and the dimensions of $\Lambda_A$, $\Lambda_B$, and $\Lambda_A \cup \Lambda_B$ are either
\begin{enumerate}
\item $\Lambda_A = \Lambda_B = \ell(s) \times \ell(s) \times \ell(s)$, $\Lambda_A \cup \Lambda_B = \ell(s) \times \ell(s) \times 2\ell(s)$,
\item $\Lambda_A = \Lambda_B = \ell(s) \times \ell(s) \times 2\ell(s)$, $\Lambda_A \cup \Lambda_B = \ell(s) \times 2\ell(s) \times 2\ell(s)$, or
\item $\Lambda_A = \Lambda_B = \ell(s) \times 2\ell(s) \times 2\ell(s)$, $\Lambda_A \cup \Lambda_B = 2\ell(s) \times 2\ell(s) \times 2\ell(s)$.
\end{enumerate}

\item $\M(A)$ and $\M(B)$ are the functions that indicate how many particles among $x_1, x_2, \cdots, x_n$ are in $\Lambda_A$ and $\Lambda_B$, respectively, when $x_{n+1}, x_{n+2}, \cdots, x_N$ are outside $(\Lambda_A \cup \Lambda_B)$, i.e.,
\begin{equation}
\M(A) := \sum_{i=1}^n 1(x_i \in \Lambda_A), \;\;\; \M(B) := \sum_{i=1}^n 1(x_i \in \Lambda_B)
\end{equation}
Note that $\M(A)$ and $\M(B)$ depend on $n$ though we omitted it.

\end{itemize}

$\Lambda_A$ are $\Lambda_B$ are a pair of boxes at the $s$-th step. Though we consider only two boxes at a time, note that we have $2^{3h-s-1}$ such pairs of boxes in $s$-th step.
\par Note that the potential term in Lemma \ref{small cell estimate}, i.e., the right hand side of \eqref{potential term in small cell estimate} is $4 \pi a |B|^{-1} f(n)$, which is different from \eqref{potential form}. $f(n)$ changes from quadratic to linear at $n = 2 \rho \ell_1^3$. At the $s$-th step, we use $f_s$ instead of $f$, and $f_s (n)$ becomes linear when $n \geq  K_s$, i.e.,

\begin{itemize}

\item Define ($1\leq s\leq 3h+1$)
\begin{equation}
f_s (t) :=
	\begin{cases}
	t(t-1) & \text{ if } t \leq K_s \\
	(2 K_s - 1) t - K_s^2 & \text{ if } t > K_s
	\end{cases}.
\end{equation}

\end{itemize}

The definition of $f_s$ ensures that $f_s$ is continuous and convex. And we define the parameters $K_s$ as follows,

\begin{itemize}

\item Let
\begin{equation}
K_1 := 2 \rho \ell_1^3
\end{equation}
and we choose $K_s$ such that
\begin{equation} \label{definition of K}
2K_s - K_{s+1} \gg |\log \rho|^{\frac{1}{2}} \sqrt{K_{s+1}} \; \text{ and } \; K_s > 2^{s-1} \rho \ell_1^3.
\end{equation}
For example,
\begin{equation}
K_s = \big( 2 - (1- \frac{1}{2^{\frac{s-1}{2}}}) |\log \rho|^{-\eta} \big) \cdot 2^{s-1} \rho \ell_1^3.
\end{equation}

\end{itemize}

We note that $f_s$ with a suitable coefficient is our actual potential energy term in \eqref{potential form}. When $s=1$, $f_s$ is equal to $f$ in \eqref{definition of f}, and $f_s (t) = t(t-1)$ when $t \leq \rho \cdot 2^s \rho \ell_1^3$, i.e., the density is no more than $\rho$. (Note that $2^{s-1} \rho \ell_1^3$ is the volume of each box in the $s$-th step.)

Our proof of Lemma \ref{large cell estimate} requires the following proposition only, where we consider $s$ to be fixed:

\begin{prop} \label{box doubling}
Let $n$ be an integer and $1\leq n\leq N$. Assume that $x_{n+1}, x_{n+2}, \cdots, x_N$ are fixed outside $(\Lambda_A \cup \Lambda_B)$. Then, for any $\psi(\x_n) \in L^2 ((\Lambda_A \cup \Lambda_B)^n)$ with $\phi_j = W_j^{-1} \psi$ for $j=1, 2, \cdots, n$ (Here $\x_n = (x_1, x_2, \cdots, x_n)$),
\begin{eqnarray}
&& \frac{(4 \pi a)^{-1}}{|\log \rho|} \sum_{j=1}^n \int_{(\Lambda_A \cup \Lambda_B)^n} |W_j|^2 |\nabla_j \phi_j(\x_n)|^2 d\x_n + \int_{(\Lambda_A \cup \Lambda_B)^n} \frac{f_s (\M(A)) + f_s (\M(B))}{|\Lambda_A|} |\psi(\x_n)|^2 d\x_n \nonumber \\
&\geq& (1 - C \rho^{\frac{1}{3}}) \int_{(\Lambda_A \cup \Lambda_B)^n} \frac{f_{s+1} (n)}{|\Lambda_A \cup \Lambda_B|} |\psi(\x_n)|^2 d\x_n. \label{box doubling result}
\end{eqnarray}
Here, the constant $C$ does not depend on $s$.
\end{prop}

Note that the factor $(4 \pi a)^{-1}$ can be a general constant $C$ in this proposition and subsequent lemmas. We keep this factor, however, in order to use it in the proof of Lemma \ref{large cell estimate}.

Proposition \ref{box doubling} shows the outcome of box doubling method when it is applied to $\Lambda_A$ and $\Lambda_B$ at the $s$-th step. We can prove Lemma \ref{large cell estimate} from this proposition.

\begin{proof} [Proof of Lemma \ref{large cell estimate}]
Recall that we defined $h$ as $\ell_2 = 2^h \ell_1$. Recall also that we have a cubic box $\Lambda_{\ell_2}$ whose side length is $\ell_2$ and $B_1, B_2, \cdots, B_{2^{3h}}$ are small cubic cells of side length $\ell_1$ such that $\bigcup_{k=1}^{2^{3h}} B_k = \Lambda_{\ell_2}$. We have that $x_{n+1}, x_{n+2}, \cdots, x_N$ are fixed outside $\Lambda_{\ell_2}$ and we let $\x_n = (x_1, x_2, \cdots, x_n)$. We want to prove that, for a given $\psi(\x_n) \in L^2 (\Lambda_{\ell_2}^n)$ with $\phi_j = W_j^{-1} \psi$ for $j=1, 2, \cdots, n$,
\begin{eqnarray}
&& \big( \frac{1}{2} - \epsilon \big) \sum_{j=1}^n \int_{\Lambda_{\ell_2}^n} |W_j|^2 |\nabla_j \phi_j(\x_n)|^2 d\x_n + \frac{4 \pi a}{\ell_1^3} \sum_{k=1}^{2^{3h}} \int_{\Lambda_{\ell_2}^n} f(\N(B_k)) |\psi(\x_n)|^2 d\x_n \nonumber \\
&\geq& \widetilde{f}(n) (1 - C \rho^{\frac{1}{3}} |\log \rho|) \frac{4 \pi a}{\ell_2^3} \| \psi \|_2^2.
\end{eqnarray}
Here,
\begin{equation}
\widetilde{f}(t) :=
	\begin{cases}
	t(t-1) & \text{ if } t \leq \rho \ell_2^3 \\
	(2 \rho \ell_2^3 - 1) t - (\rho \ell_2^3)^2 & \text{ if } t > \rho \ell_2^3
	\end{cases}.
\end{equation}

Since 
\begin{equation}
3h = \frac{\log (\ell_2^3 / \ell_1^3)}{\log 2} \leq (\frac{1}{2} - \epsilon) |\log \rho|,
\end{equation}
we can keep using Proposition \ref{box doubling} until we have only one box and its side length is $\ell_2$. Then, we get
\begin{eqnarray}
&& \big( \frac{1}{2} - \epsilon \big) \sum_{j=1}^n \int_{\Lambda_{\ell_2}} |W_j|^2 |\nabla_j \phi_j(\x_n)|^2 d\x_n + 4 \pi a \sum_{k=1}^{2^{3h}} \int_{\Lambda_{\ell_2}^n} \frac{f (\N(B_k))}{\ell_1^3} |\psi(\x_n)|^2 d\x_n \nonumber \\
&\geq& (1 - C \rho^{\frac{1}{3}} |\log \rho|) 4 \pi a \int_{\Lambda_{\ell_2}^n} \frac{f_{3h+1}(n)}{\ell_2^3} |\psi(\x_n)|^2 d\x_n.
\end{eqnarray}
Note that we used an idea similar to \eqref{division of kinetic energy} in order to convert the sum of kinetic energies over small cubic cells into the kinetic energy in a larger cubic cell.

By definition \eqref{definition of K}, $K_{3h+1} > 2^{3h} \rho \ell_1^3 = \rho \ell_2^3.$ Together with the definition of $f$ in \eqref{definition of f}, it can be easily checked that $f_{3h+1} \geq \widetilde{f}$. This proves the desired result, \eqref{result of large cell estimate}.
\end{proof}

\subsection{Proof of Proposition \ref{box doubling}}

To prove Proposition \ref{box doubling}, we consider large $n$ case and small $n$ case separately. When $n$ is large, we use the following lemma to prove Proposition \ref{box doubling}.

\begin{lem} \label{box doubling - atypical}
Suppose that $n \geq K_{s+1} + 2 \sqrt{K_{s+1}}$. Let $m_A, m_B = 0, 1, 2, \cdots, n$. Then,
\begin{eqnarray}
\min_{m_A + m_B=n} \frac{f_s (m_A) + f_s (m_B)}{|\Lambda_A|} \geq \frac{f_{s+1}(n)}{|\Lambda_A \cup \Lambda_B|}.
\end{eqnarray}
\end{lem}
We note this inequality implies \eqref{box doubling result} directly, since the kinetic energy part is always non-negative.  A proof of Lemma \ref{box doubling - atypical} will be given in the next subsection.  
\par On the other hand, for small $n$ case, i.e., when $n < K_{s+1} + 2 \sqrt{K_{s+1}}$, we are going to prove the following inequality:
\begin{eqnarray}
&& \frac{(4 \pi a)^{-1}}{|\log \rho|} \sum_{j=1}^n \int_{(\Lambda_A \cup \Lambda_B)^n} |W_j|^2 |\nabla_j \phi_j(\x_n)|^2 d\x_n + \int_{(\Lambda_A \cup \Lambda_B)^n} \frac{f_s (\M(A)) + f_s (\M(B))}{|\Lambda_A|} |\psi(\x_n)|^2 d\x_n \nonumber \\\label{key inequality for typical case}
&\geq& (1 - C \rho^{\frac{1}{3}}) F_s (0, 0, n) \|\psi \|_2^2 \\\nonumber
&\geq& (1 - C \rho^{\frac{1}{3}}) \frac{f_{s+1} (n)}{|\Lambda_A \cup \Lambda_B|} \|\psi \|_2^2. 
\end{eqnarray}
Here, $F_s$ is a function defined as follows:

\begin{defn} \label{definition of F}
For non-negative integers $n_A, n_B$, define
\begin{equation} \label{definition of F 0}
F_s (n_A, n_B, 0) := \frac{f_s(n_A)}{|\Lambda_A|} + \frac{f_s(n_B)}{|\Lambda_B|}.
\end{equation}
To define $F(n_A, n_B, k)$, we use the following process:
For a fixed $k$, consider $k$ particles, $y_1, y_2, \cdots, y_k$, which are uniformly distributed in $\Lambda_A \cup \Lambda_B$, i.e.,
\begin{equation}
P(y_i \in \Lambda_A) = P(y_i \in \Lambda_B) = \frac{1}{2}.
\end{equation}
Let $m_A (k)$ and $m_B (k)$ be the number of $y_i$'s in $\Lambda_A$ and $\Lambda_B$, respectively, i.e.,
\begin{equation}
m_A (k) = \sum_{i=1}^{k} 1(y_i \in \Lambda_A), \;\;\; m_B (k) = \sum_{i=1}^{k} 1(y_i \in \Lambda_B)
\end{equation}
Extend the definition of $F$ so that
\begin{equation} \label{definition of F k}
F_s (n_A, n_B, k) := \langle F_s (n_A + m_A (k), n_B + m_B (k), 0) \rangle_k,
\end{equation}
where $\langle \cdot \rangle_k$ denotes expectation with respect to the distribution of $y_1, y_2, \cdots, y_k$. (We call those imaginary particles $y_1, y_2, \cdots, y_k$ `randomized.')
\end{defn}

To prove \eqref{key inequality for typical case}, we need the following lemmas that will be proved in subsection 4.3:

\begin{lem} \label{F estimate}
Suppose that $n < K_{s+1} + 2 \sqrt{K_{s+1}}$. Then,
\begin{equation}
F_s (0, 0, n) \geq (1- \rho) \frac{f_{s+1}(n)}{|\Lambda_A \cup \Lambda_B|}.
\end{equation}
\end{lem}
We note this lemma implies the second inequality of \eqref{key inequality for typical case}. For the first one we have:

\begin{lem} \label{box doubling - n particles}
Suppose that $n < K_{s+1} + 2 \sqrt{K_{s+1}}$. Fix $x_{n+1}, x_{n+2}, \cdots, x_N$ outside $(\Lambda_A \cup \Lambda_B)$ and let $\x_n = (x_1, x_2, \cdots, x_n)$. Then, for any $\psi(\x_n) \in L^2 ((\Lambda_A \cup \Lambda_B)^n)$ with $\phi_j = W_j^{-1} \psi$ for $j=1, 2, \cdots, n$,
\begin{eqnarray}
&& \frac{(4 \pi a)^{-1}}{|\log \rho|} \sum_{j=1}^n \int_{(\Lambda_A \cup \Lambda_B)^n} |W_j|^2 |\nabla_j \phi_j(\x_n)|^2 d\x_n + \int_{(\Lambda_A \cup \Lambda_B)^n} F_s (\M(A), \M(B), 0) |\psi(\x_n)|^2 d\x_n \nonumber \\
&\geq& (1 - C \rho^{\frac{1}{3}}) F_s (0, 0, n) \| \psi \|_2^2.
\end{eqnarray}
\end{lem}

Now we are ready to prove Proposition \ref{box doubling}.
\begin{proof}[Proof of Proposition \ref{box doubling}]
When $n \geq K_{s+1} + 2 \sqrt{K_{s+1}}$, the desired result \eqref{box doubling result} follows from Lemma \ref{box doubling - atypical}. When $n < K_{s+1} + 2 \sqrt{K_{s+1}}$, we obtain \eqref{key inequality for typical case} from Lemma \ref{box doubling - n particles} and Lemma \ref{F estimate}, which implies \eqref{box doubling result}.
\end{proof}

\subsection{Proof of Lemma \ref{box doubling - atypical} and Lemma \ref{F estimate}} \label{4.3}

\begin{proof}[Proof of Lemma \ref{box doubling - atypical}]
Since $f_s$ is convex,
\begin{equation}
\min_{m_A + m_B=n} f_s (m_A) + f_s (m_B)\geq 2 f_s (\frac{n}{2}).
\end{equation}
Thus, it suffices to prove 
\begin{equation}\label{2fs12fs+1}
2 f_s (\frac{n}{2}) \geq \frac{1}{2} f_{s+1}(n).
\end{equation}

\begin{enumerate}

\item When $K_{s+1} + 2 \sqrt{K_{s+1}} \leq n \leq 2 K_s$, we have
\begin{eqnarray}
&& 4 f_s (\frac{n}{2}) - f_{s+1} (n) = (n^2 - 2n) - \big[ (2K_{s+1}-1) n - K_{s+1}^2 \big] \nonumber \\
&=& \big( n - K_{s+1} - \frac{1}{2} \big)^2 - K_{s+1} - \frac{1}{4} 
\end{eqnarray}
With $n-K_{s+1}\geq 2 \sqrt{K_{s+1}}$, we have it is above zero.

\item When $n > 2 K_s$: we compare the derivatives of both sides of \eqref{2fs12fs+1},
\begin{equation}
2 \frac{d}{dn} \big( f_s ({\frac{n}{2}}) \big) = 2K_s - 1 \geq K_{s+1} - \frac{1}{2} = \frac{1}{2} \frac{d}{dn} \big( f_{s+1} (n) \big).
\end{equation}
Since we also have
\begin{equation}
2 f_s (K_s) > \frac{1}{2} f_{s+1} (2K_s),
\end{equation}
We can see that, for any $n > 2 K_s$,
\begin{equation}
2 f_s (\frac{n}{2}) \geq \frac{1}{2} f_{s+1}(n).
\end{equation}

\end{enumerate}

Hence, we can get the desired lemma from cases (1) and (2).
\end{proof}

\begin{proof}[Proof of Lemma \ref{F estimate}]
Suppose that we have $n$ randomized particles $y_1, y_2, \cdots, y_n$. Each particle is uniformly distributed in $\Lambda_A \cup \Lambda_B$ so that, for $i=1, 2, \cdots, n$,
\begin{equation}
P(y_i \in \Lambda_A) = P(y_i \in \Lambda_B) = \frac{1}{2}.
\end{equation}
Let
\begin{equation}
m_A = \sum_{i=1}^n 1(y_i \in \Lambda_A), \;\;\; m_B = \sum_{i=1}^k 1(y_i \in \Lambda_B),
\end{equation}
and $\langle \cdot \rangle$ denote expectation with respect to the distribution of $y_1, y_2, \cdots, y_n$. Then,
\begin{equation}
F_s (0, 0, n) = \frac{\langle f_s (m_A) + f_s (m_B) \rangle}{|\Lambda_A|}.
\end{equation}

To compute $\langle f_s (m_A) \rangle$, we first calculate $\langle m_A^2 - m_A \rangle$ and estimate the difference $\langle m_A^2 - m_A - f_s (m_A) \rangle$. The former is
\begin{eqnarray}
\langle m_A^2 - m_A \rangle = \langle m_A \rangle^2 + (\langle m_A^2 \rangle - \langle m_A \rangle^2 ) - \langle m_A \rangle = \frac{n^2}{4} + \frac{n}{4} - \frac{n}{2} = \frac{n^2}{4} - \frac{n}{4}.
\end{eqnarray}

To estimate the latter, we use the Chernoff bound (See Corollary 4.9 in \cite{MU}.) for binomial distribution, which becomes, in this case,
\begin{equation} \label{chernoff bound}
P( m_A \geq n/2 + \zeta ) \leq e^{-\frac{2\zeta^2}{n}}.
\end{equation}
where we used the mean and variance of $m_A$ are $n/2$ and $n/4$. Since
\begin{equation}
t^2 - t - f(x) = 
	\begin{cases}
	(t-K_s)^2 & \text{ if } t > K_s \\
	0 & \text{ if } t \leq K_s
	\end{cases},
\end{equation}
\eqref{chernoff bound} implies
\begin{eqnarray}
&& \langle m_A^2 - m_A - f_s (m_A) \rangle = \sum_{\zeta=K_s}^{\infty} P(m_A = \zeta) \cdot (\zeta - K_s)^2 \leq \sum_{\zeta=K_s}^{\infty} P(m_A \geq \zeta) \cdot (\zeta - K_s)^2 \nonumber \\
&\leq& \sum_{\zeta=K_s}^{\infty} e^{-\frac{2(\zeta - n/2)^2}{n}} (\zeta - K_s)^2 \leq \int_{K_s}^{\infty} e^{-\frac{2(\zeta - n/2)^2}{n}} (\zeta - K_s)^2 d\zeta \leq C n^2 e^{-\frac{(K_s - n/2)^2}{n}}. \label{chernoff bound error estimate}
\end{eqnarray}
To estimate the last term in the inequality above, we use $n < K_{s+1} + 2\sqrt{K_{s+1}}$ and \eqref{definition of K} and obtain
\begin{eqnarray}
K_s - \frac{n}{2} &\geq& K_s - \frac{K_{s+1}}{2} - \sqrt{K_{s+1}} \gg |\log \rho|^{\frac{1}{2}} \sqrt{K_{s+1}}.
\end{eqnarray}
Hence,
\begin{equation}
\frac{(K_s - n/2)^2}{n} \geq \frac{(K_s - n/2)^2}{2^{s+2} \rho \ell_1^3} \gg |\log \rho|,
\end{equation}
which gives
\begin{equation}
n^2 e^{-\frac{(K_s - n/2)^2}{n}} \ll \rho n^2.
\end{equation}
Together with \eqref{chernoff bound error estimate} we get
\begin{equation}
\langle m_A^2 - m_A - f_s (m_A) \rangle \ll \rho n^2,
\end{equation}
thus,
\begin{equation}
\langle f_s (m_A) \rangle \geq (1 - \rho) (\frac{n^2}{4} - \frac{n}{4}).
\end{equation}

Therefore, from that 
\begin{equation}
\langle f_s (m_A) + f_s (m_B) \rangle = 2 \langle f_s (m_A) \rangle,
\end{equation}
we obtain, when $n < K_{s+1} + 2 \sqrt{K_{s+1}}$,
\begin{eqnarray}
F_s (0, 0, n) \geq \frac{1}{2} (1 - \rho) \frac{n^2 - n}{|\Lambda_A|} = (1- \rho) \frac{n^2 - n}{|\Lambda_A \cup \Lambda_B|} \geq (1- \rho) \frac{f_{s+1}(n)}{|\Lambda_A \cup \Lambda_B|} ,
\end{eqnarray}
which was to be proved.
\end{proof}

\subsection{Proof of Lemma \ref{box doubling - n particles}}

In this subsection, we consider $s$ to be fixed and let $\ell = \ell(s)$, $F = F_s$ and $n < K_{s+1} + 2 \sqrt{K_{s+1}}$. We will see that the ground state energy is equal to the expectation with respect to the uniform distribution of particles up to a small error. And this is guaranteed by the small portion of the kinetic energies $\sum T_j$. The following lemmas show the idea:

\begin{lem} \label{randomizing - one particle}When $n < K_{s+1} + 2 \sqrt{K_{s+1}}$, for any fixed $x_2, x_3, \cdots, x_N$ with $x_{n+1}, x_{n+2}, \cdots, x_N$ outside $(\Lambda_A \cup \Lambda_B)$, assume $n_A$ particles among $x_2, \cdots, x_n $ are in $\Lambda_A$ and $n_B$ particles in $\Lambda_B$($n_A+n_B=n-1$). So, when $x_1$ is in $\Lambda_A$ the total energy is $F(n_A + 1, n_B, 0)$, otherwise it is $F(n_A , n_B+1, 0)$. Let
\begin{equation} \label{definition of T1}
T_1 = - W_1^{-1} \nabla_1 W_1^2 \nabla_1 W_1^{-1}
\end{equation}
defined on functions in $L^2 (\Lambda_A \cup \Lambda_B)$ with the $W_1$-Neumann boundary conditions in $\Lambda_A \cup \Lambda_B$. Then, there exists a constant $C'$ such that, we have the following operator inequality:

\begin{eqnarray}
&& \frac{(4 \pi a)^{-1}}{|\log \rho|} T_1 + F(n_A + 1, n_B, 0) \cdot 1(x_1 \in \Lambda_A) + F(n_A, n_B + 1, 0) \cdot 1(x_1 \in \Lambda_B) \nonumber \\
&\geq& F(n_A, n_B, 1) - C' \frac{ (n_A - n_B)^2 }{\ell^4 |\log \rho|^{-1} } - C n \rho^{\frac{1}{3}} \ell^{-3}. \label{randomizing result}
\end{eqnarray}
\end{lem}

We note that the LHS of \eqref{randomizing result} depends on $x_1$, but the RHS of \eqref{randomizing result} does not. Instead, the function $F$ in the right hand side has $1$ instead of $0$ in its third argument, which means that, when we compute $F$, we need to consider one imaginary particle whose distribution is uniform in $\Lambda_A \cup \Lambda_B$ as in the definition of $F$. Thus, we can say that the particle $x_1$ got randomized. 

We note that the right hand side has a term $- C' (n_A - n_B)^2 \ell^{-4} |\log \rho|$. When we apply perturbation theory with $x_k$, $2 \leq k \leq n$, we also need to take this term into consideration. The following lemma shows an outcome of perturbation theory for general $k$:

\begin{lem} \label{randomizing - n particles}
With the assumption as above, except that $n_A$ particles among $x_{k+1}, \cdots, x_n $ are in $\Lambda_A$ and $n_B$ particles in $\Lambda_B$ ($n_A+n_B=n-k$), we have
\begin{eqnarray}
&& \frac{(4 \pi a)^{-1}}{|\log \rho|} T_k + \left( F(n_A + 1, n_B, k-1) - C' (k-1) \frac{(n_A + 1 - n_B)^2}{\ell^4 |\log \rho|^{-1} }\right) \cdot 1(x_k \in \Lambda_A) \nonumber \\
&& \hskip50pt + \left( F(n_A, n_B + 1, k-1) - C' (k-1) \frac{(n_A - n_B - 1)^2}{\ell^4 |\log \rho|^{-1} }\right) \cdot 1(x_k \in \Lambda_B) \\
&\geq& F(n_A, n_B, k) - C' k \frac{(n_A - n_B)^2}{\ell^4 |\log \rho|^{-1} } - C n \rho^{\frac{1}{3}} \ell^{-3}. \nonumber
\end{eqnarray}
Here, $C$ does not depend on $k$ and $T_k$ is defined as 
\begin{equation} \label{definition of T}
T_k = - W_k^{-1} \nabla_k W_k^2 \nabla_k W_k^{-1}
\end{equation}
\end{lem}

We can prove Lemma \ref{box doubling - n particles} using the above lemma $n$ times.

\begin{proof}[Proof of Lemma \ref{box doubling - n particles}]
Let
\begin{equation}
\M^{(k)}(A) = \sum_{i=k+1}^{n} 1(x_i \in \Lambda_A), \;\;\; \M^{(k)}(B) = \sum_{i=k+1}^{n} 1(x_i \in \Lambda_B).
\end{equation}
Note that $\M(A) = \M^{(0)}(A)$ and $\M(B) = \M^{(0)}(B)$.

From Lemma \ref{randomizing - one particle}, we get
\begin{eqnarray}
&& \frac{(4 \pi a)^{-1}}{|\log \rho|} \int_{(\Lambda_A \cup \Lambda_B)^n} |W_1|^2 |\nabla_1 \phi_1(\x_n)|^2 d\x_n + \int_{(\Lambda_A \cup \Lambda_B)^n} F(\M(A), \M(B), 0) |\psi(\x_n)|^2 d\x_n \\
&\geq& \int_{(\Lambda_A \cup \Lambda_B)^n} \big[ F(\M^{(1)}(A), \M^{(1)}(B), 1) - C' \frac{(\M^{(1)}(A)- \M^{(1)}(B))^2}{\ell^4 |\log \rho|^{-1} } - C n \rho^{\frac{1}{3}} \ell^{-3} \big] |\psi(\x_n)|^2 d\x_n. \nonumber
\end{eqnarray}
As a next step, we apply Lemma \ref{randomizing - n particles} with the right hand side of the above equation. Then we get,
\begin{eqnarray}
&& \frac{(4 \pi a)^{-1}}{|\log \rho|} \int_{(\Lambda_A \cup \Lambda_B)^n} |W_2|^2 |\nabla_2 \phi_2(\x_n)|^2 d\x_n \nonumber \\
&& + \int_{(\Lambda_A \cup \Lambda_B)^n} \big[ F(\M^{(1)}(A), \M^{(1)}(B), 1) - C' \frac{(\M^{(1)}(A)- \M^{(1)}(B))^2}{\ell^4 |\log \rho|^{-1} } \big] |\psi(\x_n)|^2 d\x_n \\
&\geq& \int_{(\Lambda_A \cup \Lambda_B)^n} \big[ F(\M^{(2)}(A), \M^{(2)}(B), 2) - C' \frac{(\M^{(2)}(A)- \M^{(2)}(B))^2}{\ell^4 |\log \rho|^{-1} } - C n \rho^{\frac{1}{3}} \ell^{-3} \big] |\psi(\x_n)|^2 d\x_n. \nonumber
\end{eqnarray}
Thus,
\begin{eqnarray}\nonumber
&& \sum_{i=1,2}\frac{(4 \pi a)^{-1}}{|\log \rho|} \int_{(\Lambda_A \cup \Lambda_B)^n} |W_i|^2 |\nabla_i \phi_i(\x_n)|^2 d\x_n+ \int_{(\Lambda_A \cup \Lambda_B)^n} F(\M(A), \M(B), 0) |\psi(\x_n)|^2 d\x_n \\
&\geq& \int_{(\Lambda_A \cup \Lambda_B)^n} \big[ F(\M^{(2)}(A), \M^{(2)}(B), 2) - C' \frac{(\M^{(2)}(A)- \M^{(2)}(B))^2}{\ell^4 |\log \rho|^{-1} } - C n \rho^{\frac{1}{3}} \ell^{-3} \big] |\psi(\x_n)|^2 d\x_n. \nonumber
\end{eqnarray}

We keep applying Lemma \ref{randomizing - n particles}. Since $\M^{(k)}(A) = \M^{(k)}(B) = 0$, when $k=n$, we get
\begin{eqnarray}
&& \sum_{i=1}^n\frac{(4 \pi a)^{-1}}{|\log \rho|} \int_{(\Lambda_A \cup \Lambda_B)^n} |W_i|^2 |\nabla_i \phi_i(\x_n)|^2 d\x_n + \int_{(\Lambda_A \cup \Lambda_B)^n} F(\M(A), \M(B), 0) |\psi(\x_n)|^2 d\x_n \nonumber \\
&\geq& \int_{(\Lambda_A \cup \Lambda_B)^n} \big[ F(\M^{(n)}(A), \M^{(n)}(B), n) - C' \frac{(\M^{(n)}(A)- \M^{(n)}(B))^2}{\ell^4 |\log \rho|^{-1} } - C n^2 \rho^{\frac{1}{3}} \ell^{-3} \big] |\psi(\x_n)|^2 d\x_n \nonumber \\
&=& \big[ F(0, 0, n) - C n^2 \rho^{\frac{1}{3}} \ell^{-3} \big] \|\psi \|_2^2.
\end{eqnarray}

Since $F$ is convex, from Jensen's inequality,
\begin{equation}
F(0, 0, n) \geq F(\frac{n}{2}, \frac{n}{2}, 0) \geq C n^2 \ell^{-3}.
\end{equation}
Thus, we have
\begin{eqnarray}
&& \frac{(4 \pi a)^{-1}}{|\log \rho|} \sum_{j=1}^n \int_{(\Lambda_A \cup \Lambda_B)^n} |W_j|^2 |\nabla_j \phi_j(\x_n)|^2 d\x_n + \int_{(\Lambda_A \cup \Lambda_B)^n} F(\M(A), \M(B), 0) |\psi(x_n)|^2 d\x_n \nonumber \\
&\geq& (1 - C \rho^{\frac{1}{3}}) F(0, 0, n) \|\psi(x_n)\|_2^2 d\x_n,
\end{eqnarray}
which was to be proved.
\end{proof}

Since Lemma \ref{randomizing - one particle} is a special case of Lemma \ref{randomizing - n particles}, we only prove Lemma \ref{randomizing - n particles}.
\begin{proof}[Proof of Lemma \ref{randomizing - n particles}]
Let
\begin{equation}
g_k (x_k) =
	\begin{cases}	
	F(n_A + 1, n_B, k-1) - F (n_A, n_B, k-1) & \text{ if } x_k \in \Lambda_A \\
	F(n_A, n_B + 1, k-1) - F (n_A, n_B, k-1) & \text{ if } x_k \in \Lambda_B
	\end{cases}
\end{equation}
and
\begin{equation}
M_k (x_k) = 
	\begin{cases}	
	(n_A + 1 - n_B)^2 - (n_A - n_B)^2 & \text{ if } x_k \in \Lambda_A \\
	(n_A - n_B - 1)^2 - (n_A - n_B)^2 & \text{ if } x_k \in \Lambda_B
	\end{cases}.
\end{equation}
Consider
\begin{equation}\label{temp4.51}
\frac{(4 \pi a)^{-1}}{|\log \rho|} T_k + g_k - C' (k-1) \frac{M_k}{\ell^4 |\log \rho|^{-1} }
\end{equation}
with $(4 \pi a |\log \rho|)^{-1} T_k$ as the unperturbed part. We want to use Temple's inequality to get a lower bound, i.e.,
\begin{eqnarray}\label{temp4.52}
&&\frac{(4 \pi a)^{-1}}{|\log \rho|} T_k + g_k - C' (k-1) \frac{M_k}{\ell^4 |\log \rho|^{-1} }\\
&\geq& F(n_A,n_B,k)-F(n_A,n_B,k-1)-\frac{C' }{\ell^4 |\log \rho|^{-1} } (n_A - n_B)^2 - C n \rho^{\frac{1}{3}} \ell^{-3}, \nonumber
\end{eqnarray}
which implies Lemma \ref{randomizing - n particles}. So to prove Lemma \ref{randomizing - n particles}, it only remains to prove \eqref{temp4.52}. 
\par In order to use Temple's inequality on \eqref{temp4.51}, we first check if the perturbation part is non-negative. This is trivial when $k=1$, so we assume that $k \geq 2$. Without loss of generality, assume $x_k \in \Lambda_A$. From definition,
\begin{eqnarray}
&& g_k (x_k) = F(n_A + 1, n_B, k-1) - F (n_A, n_B, k-1) \nonumber \\
&=& \langle F(n_A + 1 + m_A (k-1), n_B + m_B (k-1), 0) - F(n_A + m_A (k-1), n_B + m_B (k-1), 0) \rangle_{k-1} \nonumber \\
&=& \langle \frac{f_s(n_A + 1 + m_A (k-1))}{|\Lambda_A|} - \frac{f_s(n_A + m_A (k-1))}{|\Lambda_A|} \rangle_{k-1}
\end{eqnarray}
It can be easily checked that
\begin{equation}
f_s(n_A + 1 + m_A (k-1)) - f_s(n_A + m_A (k-1)) \geq C (n_A + m_A (k-1)).
\end{equation}
Thus,
\begin{eqnarray}
g_k (x_k) \geq C \ell^{-3} \langle n_A + m_A (k-1) \rangle_{k-1} \geq C k \ell^{-3}.
\end{eqnarray}
We can prove similarly that $g_k (x_k) \geq C k \ell^{-3}$ when $x_k \in \Lambda_B$. Since $n \leq C \rho \ell^3 \ll \ell |\log \rho|^{-1}$ and $M_k \leq C n$, we have
\begin{equation}
g_k \geq C k \ell^{-3} \gg C \frac{kn}{\ell^4 |\log \rho|^{-1} } \geq \frac{k M_k}{\ell^4 |\log \rho|^{-1} },
\end{equation}
which shows that the perturbation part of \eqref{temp4.51} is non-negative, i.e.,
\begin{equation}
g_k - C' k \frac{M_k}{\ell^4 |\log \rho|^{-1} } \geq 0.
\end{equation}
We also know the gap of $(4 \pi a |\log \rho|)^{-1} T_k$ is much larger than the expectation value of $g_k$ in the ground state $W_k$, since
\begin{equation}
\ell^{-2} |\log \rho|^{-1} \gg C n \ell^{-3} \geq \| g_k \|_{\infty} \geq \langle g_k \rangle_{W_k}.
\end{equation}

Hence, we can apply Temple's inequality on \eqref{temp4.51} and obtain
\begin{eqnarray}
&& \frac{(4 \pi a)^{-1}}{|\log \rho|} T_k + g_k - C' (k-1) \frac{M_k}{\ell^4 |\log \rho|^{-1} } \nonumber \\
&\geq& \langle g_k \rangle_{W_k} - C' (k-1) \langle \frac{M_k}{\ell^4 |\log \rho|^{-1} } \rangle_{W_k} - C \frac{\langle g_k^2 \rangle_{W_k} - \langle g_k \rangle_{W_k}^2}{\ell^{-2} |\log \rho|^{-1}} \label{randomizing estimate 1} \\
&& - C k^2 \frac{\langle M_k^2 \rangle_{W_k} - \langle M_k \rangle_{W_k}^2}{\ell^8 |\log \rho|^{-2} } / \big( \ell^{-2} |\log \rho|^{-1} \big) \nonumber
\end{eqnarray}
where $\langle g_k \rangle_{W_k}$ denotes
\begin{eqnarray}
\langle g_k \rangle_{W_k} = \big( \int_{\Lambda_A \cup \Lambda_B} g_k(x_k) |W_k|^2 dx_k \big) \big/ \big( \int_{\Lambda_A \cup \Lambda_B} |W_k|^2 dx_k \big)
\end{eqnarray}
and other expectations are defined similarly.

We want to estimate terms in the right hand side of \eqref{randomizing estimate 1}. In each estimate, we want to compare the expectation $\langle \cdot \rangle_{W_k}$ with an expectation with respect to a uniform distribution. Let $\langle g_k \rangle_{\textbf{1}_k}$ denote
\begin{eqnarray}
\langle g_k \rangle_{\textbf{1}_k} = \big( \int_{\Lambda_A \cup \Lambda_B} g_k(x_k) \cdot 1 \; dx_k \big) \big/ \big( \int_{\Lambda_A \cup \Lambda_B} 1 \; dx_k \big)
\end{eqnarray}
and other expectations $\langle \cdot \rangle_{{\textbf{1}_k}}$ are denoted similarly. It can be easily checked from the definitions that
\begin{equation}\label{temp4.62}
F (n_A, n_B, k-1) + \langle g_k \rangle_{\textbf{1}_k} = F (n_A, n_B, k).
\end{equation}

We have the following estimates for the terms in the right hand side of \eqref{randomizing estimate 1}.

\begin{enumerate}

\item Using Lemma \ref{expectation estimate}, we get
\begin{equation}
|\langle g_k \rangle_{W_k} - \langle g_k \rangle_{\textbf{1}_k} | \leq C n \rho \ell_0^2 \ell^{-3}.
\end{equation}
Together with \eqref{temp4.62}, we have
\begin{eqnarray}
\langle g_k \rangle_{W_k} \geq F (n_A, n_B, k)- F (n_A, n_B, k-1)- C n \rho^{\frac{1}{3}} \ell^{-3}. \label{randomizing estimate 2}
\end{eqnarray}

\item Using Lemma \ref{expectation estimate}, we get
\begin{eqnarray}
k \big| \langle \frac{M_k}{\ell^4 |\log \rho|^{-1}} \rangle_{W_k} - \langle \frac{M_k}{\ell^4 |\log \rho|^{-1} } \rangle_{\textbf{1}_k} \big| \leq C k \frac{n \rho \ell_0^2}{\ell^4 |\log \rho|^{-1}} \ll n \rho \ell_0^2 \ell^{-3}.
\end{eqnarray}
On the other hand, we simply follow the definitions to see
\begin{eqnarray}
\langle M_k  \rangle_{\textbf{1}_k} =  1.
\end{eqnarray}
Hence, we get
\begin{eqnarray}
C' (k-1) \langle \frac{M_k}{\ell^4 |\log \rho|^{-1}} \rangle_{W_k} \leq C n \rho^{\frac{1}{3}} \ell^{-3}. \label{randomizing estimate 3}
\end{eqnarray}

\item Using Lemma \ref{expectation estimate}, we get
\begin{equation}
|\langle g_k^2 \rangle_{W_k} - \langle g_k^2 \rangle_{\textbf{1}_k} | \leq C n^2 \rho \ell_0^2 \ell^{-6},
\end{equation}
\begin{equation}
|\langle g_k \rangle_{W_k}^2 - \langle g_k \rangle_{\textbf{1}_k}^2 | \leq C n^2 \rho \ell_0^2 \ell^{-6}.
\end{equation}
Thus, since $n \leq C \rho \ell^3 \ll \ell |\log \rho|^{-1}$,
\begin{eqnarray}
\frac{\langle g_k^2 \rangle_{W_k} - \langle g_k \rangle_{W_k}^2}{\ell^{-2} |\log \rho|^{-1}} \leq \frac{\langle g_k^2 \rangle_{\textbf{1}_k} - \langle g_k \rangle_{\textbf{1}_k}^2}{\ell^{-2} |\log \rho|^{-1}} + C n \rho \ell_0^2 \ell^{-3}. \label{variance g estimate 1}
\end{eqnarray}

To estimate $\langle g_k^2 \rangle_{\textbf{1}_k} - \langle g_k \rangle_{\textbf{1}_k}^2$, we again follow the definition and get
\begin{equation} \label{variance g estimate 2}
\langle g_k^2 \rangle_{\textbf{1}_k} - \langle g_k \rangle_{\textbf{1}_k}^2 = \frac{1}{4} \big[ F(n_A +1, n_B, k-1) - F(n_A, n_B + 1, k-1) \big]^2.
\end{equation}
Suppose that we have $(k-1)$ particles randomized, namely $y_1, y_2, \cdots, y_{k-1}$ as in the definition of $F$, \eqref{definition of F 0}-\eqref{definition of F k}. Then,
\begin{eqnarray}
&& \big[ F(n_A +1, n_B, k-1) - F(n_A, n_B +1, k-1) \big]^2 \\
&=& \big[ \langle F(n_A +1 + m_A (k-1), n_B + m_B (k-1), 0) - F(n_A + m_A (k-1), n_B +1 + m_B (k-1), 0) \rangle_{k-1} \big]^2 \nonumber \\
&\leq& \big \langle \big[ F(n_A +1 + m_A (k-1), n_B + m_B (k-1), 0) - F(n_A + m_A (k-1), n_B +1 + m_B (k-1), 0) \big]^2 \big \rangle_{k-1} \nonumber
\end{eqnarray}
It can be easily checked that, for any non-negative integers $m_1$ and $m_2$,
\begin{equation}
[F(m_1 + 1, m_2, 0)  - F(m_1, m_2 + 1, 0)]^2 \leq \frac{4 (m_1 - m_2)^2}{|\Lambda_A|}.
\end{equation}
Thus,
\begin{eqnarray}
&& \big[ F(n_A +1, n_B, k-1) - F(n_A, n_B +1, k-1) \big]^2 \nonumber \\
&\leq& 4 |\Lambda_A|^{-2} \big \langle \big[ (n_A + m_A (k-1)) - (n_B + m_B (k-1))^2 \big] \big \rangle_{k-1}. \label{variance g estimate 3} 
\end{eqnarray}
From the definition,
\begin{eqnarray}
&& \big \langle \big[ (n_A + m_A (k-1)) - (n_B + m_B (k-1))^2 \big] \big \rangle_{k-1}  = (n_A - n_B)^2 + (k-1). \label{variance g estimate 4}
\end{eqnarray}

Thus, from \eqref{variance g estimate 1}, \eqref{variance g estimate 2}, \eqref{variance g estimate 3}, and \eqref{variance g estimate 4}, we obtain
\begin{eqnarray}
&& \frac{\langle g_k^2 \rangle_{W_k} - \langle g_k \rangle_{W_k}^2}{\ell^{-2} |\log \rho|^{-1}} \leq |\Lambda_A|^{-2} \frac{(n_A - n_B)^2 + k}{\ell^{-2} |\log \rho|^{-1}} \leq C \frac{(n_A - n_B)^2}{\ell^4 |\log \rho|^{-1}} + C n \rho^{\frac{1}{3}} \ell^{-3}. \label{randomizing estimate 4}
\end{eqnarray}

\item Using Lemma \ref{expectation estimate}, from that $n \leq C \rho \ell^3 \ll \ell |\log \rho|^{-1}$ we get
\begin{eqnarray}
k^2 \Big| \frac{\langle M_k^2 \rangle_{W_k} - \langle M_k \rangle_{W_k^2}}{\ell^6 |\log \rho|^{-3} } - \frac{\langle M_k^2 \rangle_{\textbf{1}_k} - \langle M_k \rangle_{\textbf{1}_k}^2}{\ell^6 |\log \rho|^{-3} } \Big| \leq C k^2 \frac {n^2 \rho \ell_0^2}{\ell^6 |\log \rho|^{-3} } \ll n \rho \ell_0^2 \ell^{-3}.
\end{eqnarray}

From the definition of $M_k$, it can be easily checked that
\begin{eqnarray}
\langle M_k^2 \rangle_{\textbf{1}_k} - \langle M_k \rangle_{\textbf{1}_k}^2 = 4 (n_A - n_B)^2.
\end{eqnarray}
Thus,
\begin{eqnarray}
&& k^2 \frac{\langle M_k^2 \rangle_{\textbf{1}_k} - \langle M_k \rangle_{\textbf{1}_k}^2}{\ell^6 |\log \rho|^{-3} } = k^2 \frac{4 (n_A - n_B)^2}{\ell^6 |\log \rho|^{-3}} \leq \frac{C n^4}{\ell^6 |\log \rho|^{-3}} \ll n \rho^{\frac{1}{3}} \ell^{-3}.
\end{eqnarray}
Here, the last inequality follows from
\begin{equation}
n^3 \ell^{-3} |\log \rho|^3 \leq C \rho^3 \ell_2^6 |\log \rho|^3 \ll \rho^{\frac{1}{3}}.
\end{equation}

Hence,
\begin{eqnarray}
k^2 \frac{\langle M_k^2 \rangle_{W_k} - \langle M_k \rangle_{W_k}^2}{\ell^6 |\log \rho|^{-3} } \leq n \rho^{\frac{1}{3}} \ell^{-3}. \label{randomizing estimate 5}
\end{eqnarray}

\end{enumerate}

Applying estimates \eqref{randomizing estimate 2}, \eqref{randomizing estimate 3}, \eqref{randomizing estimate 4}, and \eqref{randomizing estimate 5} to \eqref{randomizing estimate 1}, we obtain \eqref{temp4.52}. 
This proves the desired lemma.
\end{proof}

\section{Properties of two body problem and the approximation to the ground state} \label{properties}

\begin{lem} \label{neumann problem}
Let $e_0$ and $\phi$ be the lowest Neumann eigenvalue and eigenfunction on the ball of radius $\kappa$, i.e.,
\begin{equation}
(-\Delta + \frac{1}{2}V)\phi = e_0 \phi
\end{equation}
with the boundary condition
\begin{equation}
\phi (x)=1 \text{ and  } \partial \phi (x) = 0 \text{ if } |x| = \kappa.
\end{equation}
Suppose that $V(x) = 0$ when $|x| > R_0$. Then, if $\kappa \gg R_0$, there exist constants $c_0$ and $c_1$ such that
\begin{equation}
c_0 \leq \phi(x) \leq 1, \;\;\; 1 - \phi(x) \leq \frac{c_1}{|x|}.
\end{equation}
Moreover, we have
\begin{equation} \label{definition of e_0}
\frac{3a}{\kappa^3} \leq e_0 \leq \frac{3a}{\kappa^3} ( 1 + \frac{C}{\kappa} ).
\end{equation}
\end{lem}

\begin{proof}
For a proof of the first part of the lemma,
\begin{equation}
c_0 \leq \phi(x) \leq 1, \;\;\; 1 - \phi(x) \leq \frac{c_1}{|x|},
\end{equation}
see Lemma A.1 in \cite{ESY}. To prove the second part, we extend the definition of $\phi$ so that $\phi(x) = 1$ if $|x| > \kappa$. It is well known that
\begin{equation}
\int_{\mathbb{R}^3} \phi(x) \left( -\Delta + \frac{1}{2} V(x) \right) \phi(x) dx \geq 4 \pi a.
\end{equation}
Thus, $\phi(x) \leq 1$ implies that
\begin{equation}
e_0 \geq 4 \pi a \big( \int_{|x| \leq \kappa} |\phi(x)|^2 dx \big)^{-1} \geq \frac{3a}{\kappa^3}.
\end{equation}
Upper bound for $e_0$ is also proved in Lemma A.1 in \cite{ESY}.
\end{proof}

\begin{lem} \label{properties of W}
Define $W_j$ as in \eqref{definition of W} and $\x = (x_1, x_2, \cdots, x_N)$. Suppose that $x_j$ is in $\Lambda_{\ell}$, a box of side length $\ell$ and $x_1, \cdots, \widehat{x_j}, \cdots, x_N$ are fixed.
\begin{enumerate}
\item There exists a constant $c_0 < 1$ such that
\begin{equation}
W_j (\x) \geq 1 - c_0.
\end{equation}

\item Let $\ell \gg \ell_0$. If $\Lambda_{\ell}$ contains $n$ particles including $x_j$, then
\begin{equation}
\ell^3 (1 - C n \frac{\ell_0^2}{\ell^3} - C \ell^{-1}) \leq \int_{\Lambda_{\ell}} |W_j|^2 dx_j \leq \ell^3
\end{equation}

\end{enumerate}

\end{lem}

\begin{proof}
From Lemma \ref{neumann problem}, we can see that there exists a constant $c_0$ that does not depend on $\kappa$ such that $\tau(\kappa, x_i - x_j) \leq c_0$ whenever $\ell_{-1} \leq \kappa \leq \ell_0$. It is clear that 
\begin{equation}
F_{ij}(\x) \tau(\ell_0, x_i - x_j) > 0 \text{ or } (1-F_{ij}(\x)) G_{ij}(\x) \tau(t_{ij}, x_i - x_j) > 0
\end{equation}
implies that $x_i$ is the nearest particle of $x_j$. Thus, there exists an index $k$ such that 
\begin{equation}
W_j (\x) = 1 - [F_{kj}(\x) \tau(\ell_0, x_k - x_j) + (1-F_{kj}(\x)) G_{kj}(\x) \tau(t_{kj}, x_k - x_j)],
\end{equation}
hence, we obtain the first part of the lemma.

To prove the second part of the lemma, we note that $W_j = 1$ unless there exists a particle $x_k$ such that $|x_k - x_j| < \min \{ \ell_0, t_{kj} \}$. Assume that there exists such $k$. Consider first a case where $x_k \in \Lambda_{\ell}$. When $t_{kj} > \ell_0$,
\begin{equation}
\int_{|x_k - x_j| < \ell_0} |W_j|^2 dx_j \geq \int_{|x_k - x_j| < \ell_0} \big( 1 - \frac{C}{|x_k - x_j|} \big)^2 dx_j \geq \int_{|x_k - x_j| < \ell_0} dx_j - C \ell_0^2,
\end{equation}
and, when $t_{kj} \leq \ell_0$,
\begin{equation}
\int_{|x_k - x_j| < t_{ij}} |W_j|^2 dx_j \geq \int_{|x_k - x_j| < t_{ij}} \big( 1 - \frac{C}{|x_k - x_j|} \big)^2 dx_j \geq \int_{|x_k - x_j| < t_{ij}} dx_j - C t_{ij}^2 \geq \frac{4}{3} \pi t_{ij}^3 - C \ell_0^2.
\end{equation}

Now consider the other case where $x_k \notin \Lambda_{\ell}$. Note that $d(x_k, \Lambda_{\ell}) \leq \ell_0$ in this case, where $d(x_k, \Lambda_{\ell})$ denotes the distance from $x_k$ to the box $\Lambda_{\ell}$. When $t_{kj} > \ell_0$, let $B_k$ be a ball of radius $\ell_0$ centered at $x_k$. $B_k \cap \partial \Lambda_{\ell}$, the intersection of $B_k$ and the boundary of $\Lambda_{\ell}$, is a disk, and we let $r_k$ be the radius of the disk. Since
\begin{equation}
\int_{B_k \cap \Lambda_{\ell}} \frac{1}{|x_j - x_k|} dx_j \leq \int_{\sqrt{\ell_0^2 - r_k^2} \leq |x_j - x_k| \leq \ell_0} \frac{1}{|x_j - x_k|} dx_j = 2 \pi r_k^2,
\end{equation}
we can see that
\begin{equation}
\int_{B_k \cap \Lambda_{\ell}} |W_j|^2 dx_j \geq \int_{B_k \cap \Lambda_{\ell}} \big( 1 - \frac{C}{|x_k - x_j|} \big)^2 dx_j \geq \int_{B_k \cap \Lambda_{\ell}} dx_j - C r_k^2.
\end{equation}
We can get the same estimate when $t_{kj} \leq \ell_0$, in this case, $B_k$ becomes a ball of radius $t_{kj}$ centered at $x_k$. Note that all these $B_k$'s are disjoint. If we sum over all such $k$, we get
\begin{equation}
\sum_{k: x_k \notin \Lambda_{\ell}, d(x_k, \Lambda_{\ell}) \leq \ell_0} \int_{B_k \cap \Lambda_{\ell}} |W_j|^2 dx_j \geq \sum_{k: x_k \notin \Lambda_{\ell}, d(x_k, \Lambda_{\ell}) \leq \ell_0} \int_{B_k \cap \Lambda_{\ell}} dx_j - C \sum_{k: x_k \notin \Lambda_{\ell}, d(x_k, \Lambda_{\ell}) \leq \ell_0} r_k^2.
\end{equation}
Since $B_k \cap \partial \Lambda_{\ell}$ lies on $\partial \Lambda_{\ell}$, $\sum r_k^2$ cannot exceed the area of $\partial \Lambda_{\ell}$, thus,
\begin{equation}
\sum_{k: x_k \notin \Lambda_{\ell}, d(x_k, \Lambda_{\ell}) \leq \ell_0} \int_{B_k \cap \Lambda_{\ell}} |W_j|^2 dx_j \geq \sum_{k: x_k \notin \Lambda_{\ell}, d(x_k, \Lambda_{\ell}) \leq \ell_0} \int_{B_k \cap \Lambda_{\ell}} dx_j - C \ell^2.
\end{equation}

Altogether, we get
\begin{equation}
\int_{\Lambda_{\ell}} |W_j|^2 dx_j \geq \int_{\Lambda_{\ell}} dx_j - C n \ell_0^2 - C \ell^2 = \ell^3 (1 - C n \frac{\ell_0^2}{\ell^3} - C \ell^{-1}).
\end{equation}
Finally, $W_j \leq 1$ implies that
\begin{equation}
\int_{\Lambda_{\ell}} |W_j|^2 dx_j \leq \ell^3.
\end{equation}
This proves the second part of the lemma.
\end{proof}

\begin{lem} \label{expectation estimate}
Let $\Lambda_A$ and $\Lambda_B$ be as in Section \ref{doubling a box}. Suppose $x_1, x_2, \cdots, x_n \in \Lambda_A \cup \Lambda_B$ and $n = O( \rho |\Lambda_A|)$. Let $1 \leq j \leq n$ and let $W_j$ be defined in \eqref{definition of W}. Fix $x_1, x_2, \cdots, \widehat{x_j}, \cdots, x_N$ with $x_{n+1}, x_{n+2}, \cdots, x_N$ outside $(\Lambda_A \cup \Lambda_B)$. For a given function $h(x_j)$ with
\begin{equation}
h(x_j) =
	\begin{cases}
	h_A & \text{ if } x_j \in \Lambda_A \\
	h_B & \text{ if } x_j \in \Lambda_B
	\end{cases},
\end{equation}
let
\begin{equation}
\langle h \rangle_{W_j} = \int_{\Lambda_A \cup \Lambda_B} h(x_j) |W_j|^2 dx_j \big/ \int_{\Lambda_A \cup \Lambda_B} |W_j|^2 dx_j,
\end{equation}
\begin{equation}
\langle h \rangle_{\textbf{1}_k} = \int_{\Lambda_A \cup \Lambda_B} h(x_j) \cdot 1 \; dx_j \big/ \int_{\Lambda_A \cup \Lambda_B} 1 \; dx_j.
\end{equation}
Then,
\begin{equation}
|\langle h \rangle_{W_j} - \langle h \rangle_{\textbf{1}_k}| \leq C \rho \ell_0^2 \max \{h_A, h_B \}.
\end{equation}
\end{lem}

\begin{proof}
Let
\begin{equation}
p_A = \big( \int_{\Lambda_A} |W_j|^2 dx_j \big) / \big( \int_{\Lambda_A \cup \Lambda_B} |W_j|^2 dx_j \big),
\end{equation}
\begin{equation}
p_B = \big( \int_{\Lambda_B} |W_j|^2 dx_j \big) / \big( \int_{\Lambda_A \cup \Lambda_B} |W_j|^2 dx_j \big).
\end{equation}
Then,
\begin{equation}
\langle h \rangle_{W_j} = p_A h_A + p_B h_B.
\end{equation}

From Lemma \ref{properties of W},
\begin{equation}
p_A = \frac{1}{2} + O(\frac{n \ell_0^2}{\ell^3}) + O(\ell^{-1}) = \frac{1}{2} + O(\rho \ell_0^2) + O(\ell^{-1}),
\end{equation}
and the same estimate holds for $p_B$. Since
\begin{equation}
\langle h \rangle_{\textbf{1}_k} = \frac{1}{2} h_A + \frac{1}{2} h_B,
\end{equation}
and $\ell^{-1} \ll \rho \ell_0^2$, we get
\begin{equation}
|\langle h \rangle_{W_j} - \langle h \rangle_{\textbf{1}_k}| \leq C \rho \ell_0^2 h_A + C \rho \ell_0^2 h_B \leq C \rho \ell_0^2 \max \{h_A, h_B \},
\end{equation}
which was to be proved.
\end{proof}

\section*{Acknowledgment}
We are grateful to H.-T. Yau for helpful discussions.

\end{document}